\documentclass[journal]{IEEEtran}

\usepackage{color}

\usepackage[cmex10]{amsmath}
\usepackage{amssymb}
\usepackage{amsthm}
\usepackage{algorithm}
\usepackage{algorithmic}

\usepackage{enumitem}
\usepackage{graphicx}
\usepackage{caption}

\usepackage[labelformat=simple]{subcaption}

\usepackage{tabulary}
\usepackage{multirow}

\newcommand{\defeq}{\stackrel{\triangle}{=}}

\usepackage{bm}

\long\def\comment#1{}

\newfont{\bbb}{msbm10 scaled 700}

\newfont{\bb}{msbm10 scaled 1100}

\newcommand{\cv}{{\bf c}}

\newcommand{\ev}{{\bf e}}
\newcommand{\fv}{{\bf f}}

\newcommand{\nv}{{\bf n}}

\newcommand{\tv}{{\bf t}}
\newcommand{\uv}{{\bf u}}

\newcommand{\vv}{{\bf v}}
\newcommand{\xv}{{\bf x}}
\newcommand{\yv}{{\bf y}}

\newcommand{\zerov}{{\bf 0}}
\newcommand{\onev}{{\bf 1}}

\newcommand{\Am}{{\bf A}}
\newcommand{\Bm}{{\bf B}}

\newcommand{\Dm}{{\bf D}}
\newcommand{\Em}{{\bf E}}

\newcommand{\Id}{{\bf I}}

\newcommand{\Lm}{{\bf L}}

\newcommand{\Pm}{{\bf P}}
\newcommand{\Qm}{{\bf Q}}

\newcommand{\Sm}{{\bf S}}
\newcommand{\Tm}{{\bf T}}
\newcommand{\Um}{{\bf U}}
\newcommand{\Wm}{{\bf W}}

\newcommand{\Ac}{{\cal A}}

\newcommand{\Ec}{{\cal E}}

\newcommand{\Hc}{{\cal H}}
\newcommand{\Ic}{{\cal I}}

\newcommand{\Nc}{{\cal N}}

\newcommand{\Rc}{{\cal R}}
\newcommand{\Sc}{{\cal S}}

\newcommand{\Vc}{{\cal V}}

\newcommand{\piv}{\hbox{\boldmath$\pi$}}

\newcommand{\Pim}{\hbox{\boldmath$\Pi$}}

\newcommand{\diag}{{\hbox{diag}}}

\newcommand{\trace}{{\hbox{tr}}}

\newtheorem{theorem}{Theorem}

\newtheorem{proposition}{Proposition}

\newtheorem{definition}{Definition}
\newtheorem{corollary}{Corollary}

\DeclareMathOperator*{\argmin}{arg\,min}
\DeclareMathOperator*{\argmax}{arg\,max}

\newcommand{\Scc}{{{\cal S}^c}}

\begin{document}

\title{Efficient Sampling Set Selection for Bandlimited Graph Signals Using Graph Spectral Proxies}

\author{Aamir~Anis,~\IEEEmembership{Student Member,~IEEE,}
        Akshay~Gadde,~\IEEEmembership{Student Member,~IEEE,}
        and~Antonio~Ortega,~\IEEEmembership{Fellow,~IEEE}
\thanks{A. Anis, A. Gadde and A. Ortega are with the Department
of Electrical Engineering, University of Southern California, Los Angeles,
CA, 90089 USA e-mail: (\{aanis, agadde\}@usc.edu, ortega@sipi.usc.edu).}
\thanks{This work is supported in part by NSF under grants CCF-1018977, CCF-1410009 and
CCF-1527874.}}


\maketitle

\begin{abstract}
We study the problem of selecting the best sampling set for bandlimited reconstruction of signals on graphs. 
A frequency domain representation for graph signals can be defined using the eigenvectors and eigenvalues of variation operators that 
take into account the underlying graph connectivity. 
Smoothly varying signals defined on the nodes are of particular interest in various applications, and tend to be approximately bandlimited in the frequency basis.
Sampling theory for graph signals deals with the problem of choosing the best subset of nodes for reconstructing a bandlimited signal from its samples. 
Most approaches to this problem require a computation of the frequency basis (i.e., the eigenvectors of the variation operator), followed by a search procedure using the basis elements. This can be impractical, in terms of storage and time complexity, for real datasets involving very large graphs.
We circumvent this issue in our formulation by introducing quantities called \emph{graph spectral proxies}, defined using the powers of the variation operator, in order to approximate the spectral content of graph signals. 
This allows us to formulate a direct sampling set selection approach that does not require the computation and storage of the basis elements.
We show that our approach also provides stable reconstruction when the samples are noisy or when the original signal is only approximately bandlimited. 
Furthermore, the proposed approach is valid for any choice of the variation operator, thereby covering a wide range of graphs and applications. We demonstrate its effectiveness through various numerical experiments.  
\end{abstract}


\IEEEpeerreviewmaketitle

\section{Introduction}
Graphs provide a natural representation for data in many applications, such as social networks, web information analysis, sensor networks and  machine learning~\cite{Shuman-SPM-13,Sandryhaila-SPM-14}. They can also be used to represent conventional data, such as images and videos~\cite{Narang-SSP-12,Gadde-ICIP-13}. A graph signal is a function defined over the nodes of a graph. Graph signal processing aims to extend the well-developed tools for analysis of conventional signals to signals on graphs while exploiting the underlying connectivity information~\cite{Shuman-SPM-13,Sandryhaila-SPM-14}. In this paper, we extend the theory of sampling for graph signals by developing fast and efficient algorithms for sampling set selection.

Sampling theory is of immense importance in traditional signal processing, providing a link between analog and discrete time domains and also serving as a major component in many discrete time processing systems. Fundamentally, it deals with the problem of recovering a signal from a subset of its samples. 
It provides the conditions under which the signal has a \emph{unique} and \emph{stable} reconstruction from the given samples.
Conversely, it gives the minimum sampling density required in order to get a unique and stable reconstruction for a signal that meets the modeling assumptions. 
Typically, the signal model is characterized by 
bandlimitedness in the Fourier domain. 
For example, the classical Nyquist-Shannon sampling theorem states that a signal in $L^2(\mathbb{R})$ with bandwidth $f$ can be uniquely reconstructed by its (uniformly spaced) samples if the sampling rate is higher than $2f$. Analogous results have been obtained for both regular and irregular sampling of discrete signals bandlimited in the DFT domain~\cite{Groechenig-LAA-93}.

Sampling theory of graph signals similarly deals with the problem of recovering a signal from its samples on a subset of nodes of the graph. The smoothness assumption on a graph signal is formalized in terms of bandlimitedness in a \emph{graph Fourier basis}. The graph Fourier basis is given by the eigenvectors and eigenvalues of certain \emph{variation operators} (e.g., graph Laplacian) that measure the variation in a graph signal while taking into account the underlying connectivity. To formulate a sampling theory for graph signals we need to consider the following questions: 1.~Given a subset of nodes to be sampled, what is the maximum bandwidth (in the appropriate graph Fourier domain) that a signal can have so that it can be uniquely and stably reconstructed from those samples? 2.~Given the signal bandwidth, what is the best subset of nodes to be sampled for a unique and stable reconstruction? 
Stability is an important issue in the choice of sampling set. 
In practice, signals are only approximately bandlimited and/or samples are noisy. A poor choice of sampling set can result in a very ill-conditioned reconstruction operator which amplifies the sample perturbations caused by noise and model mismatch and thus, lead to large reconstruction errors. 
Hence, selecting a sampling set that gives stable reconstructions is vital.

The problem of selecting sampling sets for recovery of smooth, bandlimited graph signals, arises in many applications. A prominent example is active semi-supervised learning~\cite{Gadde-KDD-14}, where a learner is allowed to specify a set of points to be labeled, given a budget, before predicting the unknown labels. In this setting, class indicator vectors can be considered as smooth or bandlimited graph signals and the set of points to be labeled as the sampling set. Therefore the task of actively choosing the training set in this scenario is equivalent to finding the best possible sampling set, under a given sampling budget.
Other applications of sampling set selection include selective activation of sensors in sensor networks, and design of graph based lifting transforms in image compression~\cite{Chao-PCS-15}. Signals of interest in these applications are also smooth with respect to the graph and the goal is to find the best sampling locations that  minimize the reconstruction error. 

Most recent approaches for formulating a sampling theory for graph signals involve two steps -- first, computing a portion of the graph Fourier basis, and second, using the basis elements either to check if a unique and stable reconstruction is possible with the given samples or to search for the best subset for sampling.
However, when the graphs of interest are large, computing and storing multiple eigenvectors of their variation operators increases the numerical complexity and memory requirement significantly.
Therefore, we propose a technique that achieves comparable results by using the variation operator directly and skipping the intermediate step of eigen-decomposition.


\subsection{Related work}

Sampling theory for graph signals was first studied in~\cite{Pesenson-AMS-08}, where a sufficient condition for unique recovery of signals is stated for a given sampling set. Using this condition,~\cite{Narang-ICASSP-13} gives a bound on the maximum bandwidth that a signal can have, so that it can be uniquely reconstructed from its samples on a given subset of nodes. We refine this bound in our previous work~\cite{Anis-ICASSP-14} by considering a necessary and sufficient condition for sampling. Using this condition, we also propose a direct sampling set selection method that finds a set approximately maximizing this bound so that a larger space of graph signals can be uniquely reconstructed. However, \cite{Anis-ICASSP-14} does not explain why maximizing this bound leads to stable reconstructions. Moreover, the results are specific to undirected graphs. Thus, the main contributions of this paper are to extend our prior work and propose an efficient sampling set selection algorithm that generalizes easily for different graphs, while also considering the issue of stability. 

Previous methods for sampling set selection in graphs can be classified into two types, namely spectral-domain methods and vertex-domain methods, which are summarized below. 

\subsubsection{Spectral-domain approaches}

Most of the recent work on sampling theory of graph signals assumes that a portion of the graph Fourier basis is explicitly known. We classify these methods as spectral-domain approaches since they involve computing the spectrum of the variation operator. 
For example, the work of~\cite{Shomorony-GSIP-14} requires computation and processing of the first $r$ eigenvectors of the graph Laplacian to construct a sampling set that guarantees unique (but not necessarily stable) reconstruction for a signal spanned by those eigenvectors. 
Similarly, a greedy algorithm for selecting stable sampling sets for a given bandlimited space is proposed in \cite{Chen-arxiv-15}. It considers a spectral-domain criterion, using minimum singular values of submatrices of the graph Fourier transform matrix, to minimize the effect of sample noise in the worst case. The work of~\cite{tsitsvero-arxiv-15} creates a link between the uncertainty principle for graph signals and sampling theory to arrive at similar criteria in the presence of sample noise.
It is also possible to generalize this approach using ideas from the theory of optimal experiment design~\cite{Joshi-TSP-09} and  
define other spectral-domain optimality criteria for selecting sampling sets that minimize different measures of reconstruction error when the samples are noisy (for example, the mean squared error). 
Greedy algorithms can then be used to find sets which are approximately optimal with respect to these criteria.


\subsubsection{Vertex-domain approaches}

There exist alternative approaches to sampling set selection that do not consider graph spectral information and instead rely on vertex-domain characteristics. Examples include \cite{Narang-ICIP-10} and \cite{Nguyen-TSP-15}, which select sampling sets based on maximum graph cuts and spanning trees, respectively. However, these methods are better suited for designing downsampling operators required in bipartite graph multiresolution transforms~\cite{Narang-TSP-12, Narang-TSP-13}. Specifically, they do not consider the issue of optimality of sampling sets in terms of quality of bandlimited reconstruction. 
Further, it can be shown that the maximum graph-cut based sampling set selection criterion is closely related to a special case of our proposed approach. There exists an alternate vertex-domain sampling approach, described in the work of~\cite{marques-tsp-15}, that involves successively shifting a signal using the adjacency matrix and aggregating the values of these signals on a given node. However, sampling using this strategy requires aggregating the sample values for a neighborhood size equal to the dimension of the bandlimited space, which can cover a large portion of the graph.

The sampling strategies described so far involve deterministic methods of approximating optimal sampling sets. There also exists a randomized sampling strategy~\cite{puy-arxiv-15} that guarantees a bound on the worst case reconstruction error in the presence of noise by sampling nodes independently based on a judiciously designed distribution over the nodes. 
However, one needs to sample much more nodes than the dimension of the bandlimited space to achieve the error bound.



\subsection{Contributions of this work}

It is possible to extract and process useful spectral information about a graph signal even when the graph Fourier basis is not known. For example, spectral filters in the form of polynomials of the variation operator are used in the design of wavelet filterbanks for graph signals~\cite{hammond11, Narang-TSP-12, Narang-TSP-13} to offer a trade-off between frequency-domain and vertex-domain localization. 
In our work, we use a similar technique of extracting spectral information from signals using $k$-hop localized operations, without explicitly computing the graph Fourier basis elements.
Our main contributions can be summarized as follows:
\begin{enumerate}[leftmargin=13pt]
\item Motivated by spectral filters localized in the vertex domain, we define \emph{graph spectral proxies} based on powers of the variation operator to approximate the bandwidth of graph signals. These proxies can be computed using localized operations in a distributed fashion with minimal storage cost, thus forming the key ingredient of our approach. These proxies have a tunable parameter $k$ (equal to the number of hops), that provides a trade-off between accuracy of the approximation versus the cost of computation.
\item Using these proxies, we give an approximate bound on the maximum bandwidth of graph signals (cutoff frequency) that guarantees unique reconstruction with the given samples. We show that this bound also gives us a measure of reconstruction stability for a given sampling set.
\item We finally introduce a greedy, iterative gradient based algorithm that aims to maximize the bound, in order to select an approximately optimal sampling set of given size. 

\end{enumerate}
A specific case of these spectral proxies based on the undirected graph Laplacian, and a sampling set selection algorithm, has been introduced earlier in our previous work~\cite{Anis-ICASSP-14}.
With respect to~\cite{Anis-ICASSP-14} the key new contributions are  as follows. 
We generalize the framework to a variety of  variation operators, thereby making it applicable for both undirected and directed graphs.
We provide a novel interpretation for the cutoff frequency function in~\cite{Anis-ICASSP-14} as a stability measure for a given sampling set that one can maximize. 
We show that the spectral proxies arise naturally in the expression for the bound on the reconstruction error when the samples are noisy or the signals are only approximately bandlimited. Thus, an optimal sampling set in our formulation minimizes this error bound.
We also show that our algorithm is equivalent to performing Gaussian elimination on the graph Fourier transform matrix in a certain limiting sense, and is therefore closely related to spectral-domain approaches. 
Numerical complexity of the proposed algorithm is evaluated and compared to other  state of the art methods that were introduced after \cite{Anis-ICASSP-14} was published.
Finally, we evaluate the performance and complexity of the proposed algorithm through extensive experiments using different graphs and signal models.

The rest of the paper is organized as follows. Section~\ref{sec:frequency} defines the notation used in the paper. This is followed by the concepts of frequency and bandlimitedness for signals, on both undirected and directed graphs, based on different variation operators.
In Section~\ref{sec:known_spectrum}, we consider the problems of  bandlimited reconstruction, uniqueness and stable sampling set selection, assuming that the graph Fourier basis is known. Section~\ref{sec:unknown_spectrum} addresses these problems using graph spectral proxies.
The effectiveness of our approach is demonstrated in Section~\ref{sec:experiments} through numerical experiments. We conclude in Section~\ref{sec:conclusion} with some directions for future work.


\section{Background}
\label{sec:frequency}
\subsection{Notation}

A graph $G = (\mathcal{V},\Ec)$ is a collection of nodes indexed by the set $\mathcal{V}=\{1,\ldots,N\}$ and connected by links $\Ec = \{(i,j,w_{ij})\}$, where $(i,j,w_{ij})$ denotes a link of weight $w_{ij} \in \mathbb{R}^+$ pointing from node $i$ to node $j$. 
The adjacency matrix $\mathbf{W}$ of the graph is an $N \times N$ matrix with $\Wm(i,j) = w_{ij}$. 
A graph signal is a function $f:\mathcal{V}\rightarrow \mathbb{R}$ defined on the vertices of the graph, (i.e., a scalar value assigned to each vertex). It can be represented as a vector $\mathbf{f} \in \mathbb{R}^N$ where $\fv_i$ represents the function value on the $i^\text{th}$ vertex.
For any $\xv \in \mathbb{R}^N$ and a set $\Sc \subseteq \{1,\dots,N\}$, we use $\xv_\Sc$ to denote a sub-vector of $\xv$ consisting of components indexed by $\Sc$. Similarly, for $\Am \in \mathbb{R}^{N\times N}$, $\Am_{\Sc_1\Sc_2}$ is used to denote the sub-matrix of $\Am$ with rows indexed by $\Sc_1$ and columns indexed by $\Sc_2$. For simplicity, we denote $\Am_{\Sc\Sc}$ by $\Am_\Sc$. 
The complement of $\Sc$ in $\Vc$ is denoted by $\Sc^c = \Vc \smallsetminus \Sc$. 
Further, we define $L_2(\Sc)$ to be the space of all graph signals which are zero everywhere except possibly on the subset of nodes $\Sc$, i.e.,
\begin{equation}
L_2(\Sc) = \{\xv \in \mathbb{R}^N \mid \xv_{\Sc^c} = \zerov\}.
\end{equation}


\subsection{Notions of Frequency for Graph Signals}
In order to formulate a sampling theory for graph signals, we need a notion of frequency that enables us to characterize the level of smoothness of the graph signal with respect to the graph. The key idea, which is used in practice, is to define analogs of operators such as shift or variation from traditional signal processing, that allow one to transform a signal or measure its properties while taking into account the underlying connectivity over the graph. Let $\Lm$ be such an operator in the form of an $N \times N$ matrix\footnote{Although $\Lm$ has been extensively used to denote the combinatorial Laplacian in graph theory, we overload this notation to make the point that any such variation operator can be defined to characterize signals of interest in the application at hand.}. 
A variation operator creates a notion of smoothness for graph signals through its spectrum.
Specifically, assume that $\Lm$ has eigenvalues $|\lambda_1| \leq \ldots \leq |\lambda_N|$ and corresponding eigenvectors $\{\uv^1, \ldots, \uv^N\}$.  
Then, these eigenvectors provide a Fourier-like basis for graph signals with the frequencies given by the corresponding eigenvalues.
For each $\Lm$, one can also define a variation functional $\text{Var}(\Lm,\fv)$ that measures the variation in any signal $\fv$ with respect to $\Lm$. 
Such a definition should induce an ordering of the eigenvectors which is consistent with the ordering of eigenvalues. More formally, if $|\lambda_i| \leq |\lambda_j|$, then $\text{Var}(\Lm,\uv^i) \leq \text{Var}(\Lm,\uv^j)$.

The graph Fourier transform (GFT) $\tilde{\fv}$ of a signal $\fv$ is given by its representation in the above basis, $\tilde{\fv} = \Um^{-1}\fv$, where $\Um = [\uv^1 \ldots \uv^N]$. Note that a GFT can be defined using any variation operator. Examples of possible variation operators are reviewed in Section~\ref{sec:var_ops}. 
If the variation operator $\Lm$ is symmetric then its eigenvectors are orthogonal leading to an orthogonal GFT. In some cases, $\Lm$ may not be diagonalizable. In such cases, one can resort to the Jordan normal form~\cite{Moura-TSP-14} and use generalized eigenvectors. 

A signal $\fv$ is  said to be $\omega$-bandlimited if $\tilde{\fv}_i = 0$ for all $i$ with $|\lambda_i| > \omega$. In other words, GFT of an $\omega$-bandlimited\footnote{One can also define highpass and bandpass signals in the GFT domain. Sampling theory can be generalized for such signals by treating them as lowpass in the eigenbasis of a shifted variation operator, \emph{e.g.}, one can use $\Lm' = |\lambda_N| \Id - \Lm$ for highpass signals.} $\fv$ is supported on frequencies in $[0,\omega]$. If $\{\lambda_1, \lambda_2, \ldots,\lambda_r\}$ are the eigenvalues of $\Lm$ less than or equal to $\omega$ in magnitude,
then any $\omega$-bandlimited signal can be written as a linear combination of the corresponding eigenvectors:
\begin{equation}
\fv = \sum_{i=1}^r \tilde{\fv}_i \uv^i = \Um_{\Vc\Rc} \tilde{\fv}_\Rc,
\label{eq:bl_sig}
\end{equation}  
where $\Rc = \{1,\ldots,r\}$. The space of $\omega$-bandlimited signals is called Paley-Wiener space and is denoted by $PW_\omega(G)$~\cite{Pesenson-AMS-08}. Note that $PW_\omega(G) = \text{range}(\Um_{\Vc\Rc})$ (\emph{i.e.}, the span of columns of $\Um_{\Vc\Rc}$).
Bandwidth of a signal $\fv$ is defined as the largest among absolute values of eigenvalues corresponding to non-zero GFT coefficients of $\fv$, i.e.,
\begin{equation}
\omega(\fv) \defeq \max_i \; \{|\lambda_i| \;|\; \tilde{\fv}_i \neq 0 \}.
\end{equation}

A key ingredient in our theory is an approximation of the bandwidth of a signal using powers of the variation operator $\Lm$, as explained in Section~\ref{sec:unknown_spectrum}. Since this approximation holds for any variation operator, the proposed theory remains valid for all of the choices of GFT in Table~\ref{tab:gft_bases}.


\subsection{Examples of variation operators} 
\label{sec:var_ops}

\subsubsection{Variation on undirected graphs}

In undirected graphs, the most commonly used variation operator is the combinatorial Laplacian~\cite{Chung-97} given by:
\begin{equation}
\Lm = \Dm - \Wm,
\end{equation}
where $\Dm$ is the diagonal degree matrix $\text{diag}\{d_1, \dots, d_N\}$ with $d_i = \sum_j w_{ij}$. Since, $w_{ij} = w_{ji}$ for undirected graphs, this matrix is symmetric. As a result, it has real non-negative eigenvalues $\lambda_i \geq 0$ and an orthogonal set of eigenvectors. The variation functional associated with this operator is known as the \emph{graph Laplacian quadratic form}~\cite{Shuman-SPM-13} and is given by
\begin{equation}
\text{Var}_\textit{QF}(\fv) = \fv^\top \Lm \fv = \frac{1}{2} \sum_{i,j} w_{ij} (f_i - f_j)^2.
\label{eq:var_qf}
\end{equation}
One can normalize the combinatorial Laplacian to obtain the symmetric normalized Laplacian and the (asymmetric) random walk Laplacian given as
\begin{equation}
\Lm_\textit{sym} = \Dm^{-1/2} \Lm \Dm^{-1/2}, \quad \Lm_\textit{rw} = \Dm^{-1} \Lm.
\end{equation}
Both $\Lm_\textit{sym}$ and $\Lm_\textit{rw}$ have non-negative eigenvalues. However the eigenvectors of $\Lm_\textit{rw}$ are not orthogonal as it is asymmetric. The eigenvectors of $\Lm_\textit{sym}$, on the other hand, are orthogonal. The variation functional associated with $\Lm_\textit{sym}$ has a nice interpretation as it normalizes the signal values on the nodes by the degree:
\begin{equation}
\text{Var}_\textit{QFsym}(\fv) = \fv^\top \Lm_\textit{sym} \fv = \frac{1}{2}\sum_{i,j} w_{ij} \left( \frac{f_i}{\sqrt{d_i}}-\frac{f_j}{\sqrt{d_j}} \right)^2.
\label{eq:var_qf_sym}
\end{equation}



\subsubsection[]{Variation on directed graphs}
Note that variation operators defined for directed graphs can also be used for undirected graphs since each undirected edge can be thought of as two oppositely pointing directed edges. 

\paragraph{Variation using the adjacency matrix}
This approach involves posing the adjacency matrix as a \emph{shift operator} over the graph (see \cite{Moura-TSP-14} for details). For any signal $\fv \in \mathbb{R}^n$, the signal $\Wm \fv$ is considered as a shifted version of $\fv$ over the graph, analogous to the shift operation defined in digital signal processing. Using this analogy, \cite{Moura-TSP-14} defines \emph{total variation} of a signal $\fv$ on the graph as
\begin{equation}
\label{eq:var_tv}
\text{Var}^p_{TV}(\fv) = \left\|\fv - \frac{1}{|\mu_{\text{max}}|} \Wm \fv\right\|_p,
\end{equation}
where $p = 1,2$ and $\mu_{\text{max}}$ denotes the eigenvalue of $\Wm$ with the largest magnitude. It can be shown that for two eigenvalues $|\mu_i| < |\mu_j|$ of $\Wm$, the corresponding eigenvectors $\vv_i$ and $\vv_j$ satisfy $\text{Var}^p_{TV}(\vv^i) < \text{Var}^p_{TV}(\vv^j)$. In order to be consistent with our convention, one can define the variation operator as $\Lm = \mathbf{I} - \Wm/|\mu_{\text{max}}|$ which has the same eigenvectors as $\Wm$ with eigenvalues $\lambda_i = 1-\mu_i/|\mu_{\text{max}}|$. This allows us to have the same ordering for the graph frequencies and the variations in the basis vectors. 
Note that for directed graphs, where $\Wm$ is not symmetric, the GFT basis vectors will not be orthogonal. Further, for some adjacency matrices, there may not exist a complete set of linearly independent eigenvectors. In such cases, one can use generalized eigenvectors in the Jordan normal form of $\Wm$ as stated before~\cite{Moura-TSP-14}. 

%

\begin{table*}[t]
\centering
\caption{Different choices of the variation operator $\Lm$ for defining GFT bases.}
\label{tab:gft_bases}
\begin{tabular}{p{0.6in}p{2.2in}p{0.5in}p{2in}p{1.00in}}
\hline\hline
Operator                 & Expression                                                                                                                                                                                                             & Graph type          & Associated variation functional                                                                                                      & Properties                                           \\
\hline
Combinatorial            & $\Lm = \Dm - \Wm$                                                                                                                                                    & Undirected          & $\fv^\top \Lm \fv = \frac{1}{2} \sum_{i,j} w_{ij} (\fv_i - \fv_j)^2$                                                                 & Symmetric, $\lambda_i \geq 0$, $\Um$ orthogonal      \\
\hline
Symmetric normalized     & $\Lm = \Id - \Dm^{-1/2} \Wm \Dm^{-1/2}$                                                                                                                                                           & Undirected          & $\fv^\top \Lm \fv =  \frac{1}{2}\sum_{i,j} w_{ij} \left( \frac{\fv_i}{\sqrt{d_i}}-\frac{\fv_j}{\sqrt{d_j}} \right)^2$                & Symm. , $\lambda_i \in [0,2]$, $\Um$ orthogonal  \\
\hline
Random walk (undirected) & $\Lm = \Id - \Dm^{-1} \Wm$                                                                                                                                                                       & Undirected          & $||\Lm \fv ||$                                                                                                                       & Asymmetric, $\lambda_i \geq 0$, $\Um$ non-orthogonal \\
\hline
Adjacency-based          & $\Lm = \Id - \frac{1}{|\mu_{\rm max}|} \Wm$, $\mu_{\rm max}$: maximum eigenvalue of $\Wm$                                                                                                                              & Undirected/ Directed & $|| \Lm \fv ||_p, p = 1,2$                                                                                                           & Asymm., non-orthog. $\Um$  for directed graphs, $\text{Re}\{\lambda_i\} \geq 0$  \\
\hline
Hub-authority            & $\Lm = \gamma (\Id - \Tm^\top \Tm) + (1-\gamma) (\Id - \Tm \Tm^\top)$, $\Tm = \Dm_p^{-1/2} \Wm \Dm_q^{-1/2}$, $\Dm_p = {\rm diag} \{ p_i \}, p_i = \sum_j w_{ji}$, $\Dm_q = {\rm diag} \{ q_i \}, q_i = \sum_j w_{ij}$ & Directed            & $\fv^\top \Lm\fv$, see text for complete expression.                                                            & Symmetric, $\lambda_i \geq 0$, $\Um$ orthogonal      \\
\hline
Random walk (directed)   & $\Lm = \Id - \frac{1}{2}\left(\Pim^{1/2}\Pm \Pim^{-1/2} + \Pim^{-1/2}\Pm^\top \Pim^{1/2}\right)$, $\Pm_{ij} = w_{ij}/\sum_j w_{ij}$, $\Pim = {\rm diag} \{ \pi_i \}$                                                   & Directed            & $\fv^\top \Lm \fv = \frac{1}{2}\sum_{i,j} \piv_i \Pm_{ij} \left(\frac{\fv_i}{\sqrt{\piv_i}} - \frac{\fv_j}{\sqrt{\piv_j}} \right)^2$ & Symmetric, $\lambda_i \geq 0$, $\Um$ orthogonal     \\
\hline
\end{tabular}
\end{table*}

\paragraph{Variation using the hub-authority model}

This notion of variation is based on the hub-authority model~\cite{Kleinberg-ACM-99} for specific directed graphs such as a hyperlinked environment (e.g., the web). This model distinguishes between two types of nodes. Hub nodes $\Hc$ are the subset of nodes which point to other nodes, whereas authority nodes $\Ac$ are the nodes to which other nodes point. Note that a node can be both a hub and an authority simultaneously. In a directed network, we need to define two kinds of degrees for each node $i \in \Vc$, namely the in-degree $p_i = \sum_j w_{ji}$ and the out-degree $q_i = \sum_j w_{ij}$. The co-linkage between two authorities $i,j \in \Ac$ or two hubs $i,j \in \Hc$ is defined as
\begin{equation}
c_{ij} = \sum_{h\in \Hc} \frac{w_{hi}w_{hj}}{q_h} \quad \text{and} \quad c_{ij} = \sum_{a\in \Ac} \frac{w_{ia}w_{ja}}{p_a}
\end{equation}
respectively, and can be thought of as a cumulative link weight between two authorities (or hubs). Based on this, one can define a variation functional for a signal $\fv$ on the authority nodes~\cite{Zhou-NIPS-04} as
\begin{equation}
\text{Var}_\Ac(\fv) = \frac{1}{2}\sum_{i,j \in \Ac} c_{ij} \left(\frac{f_i}{\sqrt{p_i}}-\frac{f_j}{\sqrt{p_j}}\right)^2.
\label{eq:auth_var}
\end{equation}
In order to write the above functional in a matrix form, define $\Tm = \Dm_q^{-1/2}\Wm\Dm_p^{-1/2}$, where $\Dm_p^{-1/2}$ and $\Dm_q^{-1/2}$ are diagonal matrices with 
\[
{(\Dm_p^{-1/2})}_{ii} =
\begin{cases} 
\frac{1}{\sqrt{p_i}} \text{ if } p_i \neq 0 \\
0 \; \text{ otherwise},
\end{cases}
{(\Dm_q^{-1/2})}_{ii} =
\begin{cases} 
\frac{1}{\sqrt{q_i}} \text{ if } q_i \neq 0 \\
0 \; \text{ otherwise}.
\end{cases}
\]
It is possible to show that $\text{Var}_\Ac(\fv) = \fv^\top \Lm_\Ac \fv$, where $\Lm_\Ac = \mathbf{I} - \Tm^\top \Tm$. A variation functional for a signal $\fv$ on the hub nodes can be defined in the same way as \eqref{eq:auth_var} and can be written in a matrix form as $\text{Var}_\Hc(\fv) = \fv^\top \Lm_\Hc \fv$, where $\Lm_\Hc = \mathbf{I}-\Tm \Tm^\top$. A convex combination $\text{Var}_\gamma(\fv) = \gamma \text{Var}_\Ac(\fv) + (1-\gamma) \text{Var}_\Hc(\fv)$, with $\gamma \in [0,1]$,  can be used to define a variation functional for $\fv$ on the whole vertex set $\Vc$. Note that the corresponding variation operator $\Lm_\gamma = \gamma\Lm_\Ac + (1-\gamma)\Lm_\Hc$ is symmetric and positive semi-definite. Hence, eigenvectors and eigenvalues of $\Lm_\gamma$ can be used to define an orthogonal GFT similar to the undirected case, where the variation in the eigenvector increases as the corresponding eigenvalue increases.

\paragraph{Variation using the random walk model}
Every directed graph has an associated random walk with a probability transition matrix $\Pm$ given by 
\begin{equation}
\Pm_{ij} = \frac{w_{ij}}{\sum_{j}w_{ij}}.
\end{equation}
By the Perron-Frobenius theorem, if $\Pm$ is irreducible then it has a stationary distribution $\piv$ which satisfies $\piv \Pm = \piv$~\cite{Horn-Johnson}. One can then define the following variation functional for signals on directed graphs~\cite{Chung-05, Zhou-ICML-05}:
\begin{equation}
\text{Var}_{rw}(\fv) = \frac{1}{2}\sum_{i,j} \piv_i \Pm_{ij} \left(\frac{f_i}{\sqrt{\piv_i}} - \frac{f_j}{\sqrt{\piv_j}} \right)^2.
\end{equation} 
Note that if the graph is undirected, the above expression reduces to \eqref{eq:var_qf_sym} since, in that case, $\piv_i = d_i/\sum_{j}d_j$. Intuitively, $\piv_i \Pm_{ij}$ can be thought of as the probability of transition from node $i$ to $j$ in the steady state. We expect it to be large if $i$ is similar to $j$. Thus, a big difference in signal values on nodes similar to each other contributes more to the variation. A justification for the above functional in terms of generalization of normalized cut to directed graphs is given in~\cite{Chung-05, Zhou-ICML-05}. Let $\Pim = \diag\{\piv_1,\ldots, \piv_n\}$. Then $\text{Var}_\text{rw}(\fv)$ can be written as $\fv^\top\Lm \fv$, where
\begin{equation}
\Lm = \mathbf{I} - \frac{1}{2}\left(\Pim^{1/2}\Pm \Pim^{-1/2} + \Pim^{-1/2}\Pm^\top \Pim^{1/2}\right).
\end{equation}
It is easy to see that the above $\Lm$ is a symmetric positive semi-definite matrix. Therefore, its eigenvectors can be used to define an orthonormal GFT, where the variation in the eigenvector increases as the corresponding eigenvalue increases.

Table~\ref{tab:gft_bases} summarizes different choices of GFT bases based on the above variation operators. Our theory applies to all of these choices of GFT (with the caveat that diagonalizability is assumed in the definition of adjacency-based GFT).


\section{Sampling theory for graph signals}
\label{sec:known_spectrum}
In this section, we address the issue of uniqueness and stability of bandlimited graph signal reconstruction and discuss different optimality criteria for sampling set selection assuming that the graph Fourier basis (i.e., the spectrum of the corresponding variation operator) is known.
The uniqueness conditions in this section are equivalent to the ones in ~\cite{Shomorony-GSIP-14, Chen-arxiv-15, Eldar-JFA-03}. However, the specific form in which we present these conditions lets us give a GFT-free definition of cutoff frequency. This together with the spectral proxies defined later in Section~\ref{sec:unknown_spectrum} allows us to circumvent the explicit computation of the graph Fourier basis to ensure uniqueness and find a good sampling set. 

The results in this section are useful when the graphs under consideration are small and thus, computing the spectrum of their variation operators is computationally feasible. They also serve as a guideline for tackling the aforementioned questions when the graphs are large and computation and storage of the graph Fourier basis is impractical.

\subsection{Uniqueness of Reconstruction}
In order to give a necessary and sufficient condition for unique identifiability of any signal $\fv \in PW_\omega(G)$ from its samples $\fv_\Sc$ on the sampling set $\Sc$, we first state the concept of uniqueness set~\cite{Pesenson-AMS-08}.  
\begin{definition}[Uniqueness set] A subset of nodes $\Sc$ is a uniqueness set for the space $PW_\omega(G)$ iff $\xv_\Sc = \yv_\Sc$ implies $\xv = \yv$ for all $\xv,\yv \in PW_\omega(G)$.
\end{definition}
Unique identifiability requires that no two bandlimited signals have the same samples on the sampling set as ensured by the following theorem in our previous work~\cite{Anis-ICASSP-14}.
\begin{theorem}[Unique sampling]
\label{thm:uniqueness}
$\Sc$ is a uniqueness set for $PW_\omega(G)$ if and only if $PW_\omega(G) \cap L_2(\Sc^c) = \{ \zerov \}$.
\end{theorem}
%
%
%
Let $\Sm$ be a matrix whose columns are indicator functions for nodes in $\Sc$. Note that $\Sm^\top:\mathbb{R}^n \rightarrow \mathbb{R}^{|\Sc|}$ is the sampling operator with $\Sm^\top \fv = \fv_\Sc$. Theorem~\ref{thm:uniqueness} essentially states that no signal in $PW_\omega(G)$ is in the null space $\Nc(\Sm^\top)$ of the sampling operator. 
Any $\fv \in PW_\omega(G)$ can be written as $\fv = \Um_{\Vc\Rc}\cv$.
Thus, for unique sampling of any signal in $PW_\omega(G)$ on $\Sc$, we need $\Sm^\top \Um_{\Vc\Rc}\cv = \Um_{\Sc\Rc} \cv \neq \zerov \;\; \forall \; \cv \neq \zerov$. This observation leads to the following corollary (which is also given in~\cite{Chen-ICASSP-15}).
\begin{corollary}
\label{thm:uniqueness_U}
Let $\Rc = \{1, \dots, r \}$, where $\lambda_r$ is the largest graph frequency less than $\omega$. Then 
$\Sc$ is a uniqueness set for $PW_\omega(G)$ if and only if $\Um_{\Sc\Rc}$ has full column rank.
\end{corollary}
If $\Um_{\Sc\Rc}$ has a full column rank, then a unique reconstruction $\hat{\fv} \in PW_\omega(G)$ can be obtained by finding the unique least squares solution to $\fv_\Sc = \Um_{\Sc\Rc}\cv$:
\begin{equation}
\hat{\fv} = \Um_{\Vc\Rc} \Um_{\Sc\Rc}^+ \fv_\Sc,
\label{eq:bl_recon}
\end{equation}
where $ \Um_{\Sc\Rc}^+ = (\Um_{\Sc\Rc}^\top\Um_{\Sc\Rc})^{-1}\Um_{\Sc\Rc}^\top$ is the pseudo-inverse of $\Um_{\Sc\Rc}$. The above reconstruction formula is also known as consistent reconstruction~\cite{Eldar-JFA-03} since it keeps the observed samples unchanged\footnote{Existence of a sample consistent reconstruction in $PW_\omega(G)$ requires that $PW_\omega(G) \oplus L_2(\Scc) = \mathbb{R}^N$~\cite{Eldar-JFA-03}.}, i.e., $\hat{\fv}_\Sc = \fv_\Sc$. Moreover, it is easy to see that if the original signal $\fv \in PW_\omega(G)$, then $\hat{\fv} = \fv$.
%


\begin{table*}[t]
\centering
\caption{Summary of uniqueness conditions and sampling set selection criteria when the graph Fourier basis is known}
\label{tab:known_spectrum results}
\begin{tabular}{p{1.6in}|p{1.85in}|p{1.1in}|p{1.9in}}
\hline 
Assumption & Objective & Optimality criteria & Algorithm  \\ \hline \hline

bandlimited $\fv$, noise free samples  & Unique reconstruction & full column rank $\Um_{\Sc\Rc}$ & Gaussian elimination, greedy \cite{Shomorony-GSIP-14} \\ \hline

\multirow{2}{*}{bandlimited $\fv$, noisy samples} 
&  minimum reconstruction MSE & $\min \trace[(\Um_{\Sc\Rc}^\top\Um_{\Sc\Rc})^{-1}]$ & Gaussian elimination with pivoting, \\ \cline{2-3}
& minimum worst case reconstruction error & $\min \sigma_{\min}(\Um_{\Sc\Rc})$ & greedy~\cite{Chen-arxiv-15}, convex optimization~\cite{Boyd-CVX} \\ \hline

approximately bandlimited $\fv$, noise free samples & minimum worst case reconstruction error & $\min \sigma_{\min}(\Um_{\Sc\Rc})$ & Gaussian elimination with pivoting, greedy~\cite{Chen-arxiv-15}, convex optimization~\cite{Boyd-CVX}
\\ \hline 
\end{tabular} 
\end{table*}

\subsection{Issue of Stability and Choice of Sampling set}
\label{sec:stability_intro}
Note that selecting a sampling set $\Sc$ for $PW_\omega(G)$ amounts to selecting a set of rows of $\Um_{\Vc\Rc}$. It is always possible to find a sampling set of size $r = \dim{PW_\omega(G)}$ that uniquely determines signals in $PW_\omega(G)$ as proven below. 
\begin{proposition}
For any $PW_\omega(G)$, there always exists a uniqueness set $\Sc$ of size $|\Sc| = r$. 
\end{proposition}
\begin{proof}
Since $\{\uv^i\}_{i=1}^r$ are linearly independent, the matrix $\Um_{\Vc\Rc}$ has full column rank equal to $r$. Further, since the row rank of a matrix equals its column rank, we can always find a linearly independent set $\Sc$ of $r$ rows such that $\Um_{\Sc\Rc}$ has full rank that equals $r$, thus proving our claim. 
\end{proof}
In most cases picking $r$ nodes randomly gives a full rank $\Um_{\Sc\Rc}$. However, all sampling sets of given size are not equally good. A bad choice of $\Sc$ can give an ill-conditioned $\Um_{\Sc\Rc}$ which in turn leads to an unstable reconstruction $\hat{\fv}$. 
Stability of reconstruction is important when the true signal $\fv$ is only approximately bandlimited (which is the case for most signals in practice) or when the samples are noisy. The reconstruction error in this case depends not only on noise and model mismatch but also on the choice of sampling set. The best sampling set achieves the smallest reconstruction error. 

\vspace{0.5\baselineskip}
\subsubsection{Effect of noise}
We first consider the case when the observed samples are noisy. Let $\fv \in PW_\omega(G)$ be the true signal and $\nv \in \mathbb{R}^{|\Sc|}$ be the noise introduced during sampling. The observed samples are then given by $\yv_\Sc = \fv_\Sc + \nv$. Using \eqref{eq:bl_recon}, we get the following reconstruction
\begin{equation}
\hat{\fv} = \Um_{\Vc\Rc}\Um_{\Sc\Rc}^+\fv_\Sc + \Um_{\Vc\Rc}\Um_{\Sc\Rc}^+\nv.
\end{equation}
Since $\fv \in PW_\omega(G)$, $\Um_{\Vc\Rc}\Um_{\Sc\Rc}^+\fv_\Sc = \fv$. The reconstruction error equals $\ev = \hat{\fv} - \fv = \Um_{\Vc\Rc}\Um_{\Sc\Rc}^+\nv$. If we assume that the entries of $\nv$ are iid with zero mean and unit variance, then the covariance matrix of the reconstruction error is given by
\begin{equation}
\Em = \mathbb{E}[\ev\ev^\top] = \Um_{\Vc\Rc} (\Um_{\Sc\Rc}^\top\Um_{\Sc\Rc})^{-1} \Um_{\Vc\Rc}^\top.
\end{equation}
Different costs can be defined to measure the reconstruction error as a function of the error covariance matrix. These cost functions are based on optimal design of experiments~\cite{Boyd-CVX}. If we define the optimal sampling set $\Sc^\text{opt}$ of size $m$, as the set which minimizes the mean squared error, then assuming $\Um_{\Vc\Rc}$ has orthonormal columns, we have
\begin{equation}
\Sc^\text{A-opt} = \argmin_{|\Sc| = m} \; \trace[\Em] = \argmin_{|\Sc| = m} \; \trace[(\Um_{\Sc\Rc}^\top\Um_{\Sc\Rc})^{-1}].
\label{eq:a-opt}
\end{equation}
This is analogous to the so-called $A$-optimal design. Similarly, minimizing the maximum eigenvalue of the error covariance matrix leads to $E$-optimal design. For an orthonormal $\Um_{\Vc\Rc}$, the optimal sampling set with this criterion is given by
\begin{equation}
\Sc^\text{E-opt} = \argmin_{|\Sc| = m} \; \lambda_{\max}(\Em) = \argmax_{|\Sc|=m} \; \sigma_{\min}(\Um_{\Sc\Rc}),
\label{eq:e-opt}
\end{equation}
where $\sigma_{\min}(.)$ 
denotes the smallest singular value of a matrix. It can be thought of as a sampling set which minimizes the worst case reconstruction error. The above criterion is equivalent to the one proposed in \cite{Chen-arxiv-15}. Further, one can show that when $\Um_{\Vc\Rc}$ does not have orthonormal columns, ~\eqref{eq:a-opt} and~\eqref{eq:e-opt} produce sampling sets that minimize upper bounds on the mean squared and worst case reconstruction errors respectively. Note that both $A$ and $E$-optimality criteria lead to combinatorial problems, but it is possible to develop greedy approximate solutions to these problems.

So far we assumed that the true signal $\fv \in PW_\omega(G)$ and hence, $\Um_{\Vc\Rc}\Um_{\Sc\Rc}^+\fv_\Sc = \fv$. However, in most applications, the signals are only approximately bandlimited. The reconstruction error in such a case is analyzed next.

%
\vspace{0.5\baselineskip}
\subsubsection{Effect of model mismatch}
Let $\Pm = \Um_{\Vc\Rc}\Um_{\Vc\Rc}^\top$ be the projector for $PW_\omega(G)$ and $\Qm = \Sm\Sm^\top$ be the projector for $L_2(\Sc)$. 
Assume that the true signal is given by $\fv = \fv^* + \Delta \fv$, where $\fv^* = \Pm \fv$ is the bandlimited component of the signal and $\Delta \fv = \Pm^\perp\fv$ captures the ``high-pass component'' (i.e., the model mismatch). If we use \eqref{eq:bl_recon} for reconstructing $\fv$, then a tight upper bound on the reconstruction error~\cite{Eldar-JFA-03} is given by
\begin{equation}
\|\fv - \hat{\fv}\| \leq \frac{1}{\cos(\theta_{\max})} \|\Delta\fv\|, 
\label{eq:error_bound}
\end{equation}
where $\theta_{\max}$ is the maximum angle between subspaces $PW_\omega(G)$ and $L_2(\Sc)$ defined as
\begin{equation}
\cos(\theta_{\max}) = \inf_{\fv \in PW_\omega(G), \|\fv\| = 1} \|\Qm\fv\|. 
\label{eq:theta_max}
\end{equation}
$\cos(\theta_{\max}) > 0$ when the uniqueness condition in Theorem~\ref{thm:uniqueness} is satisfied and the error is bounded. Intuitively, the above equation says that for the worst case error to be minimum, the sampling and reconstruction subspaces should be as aligned as possible.

We define an optimal sampling set $\Sc^\text{opt}$ of size $m$ for $PW_\omega(G)$  as the set which minimizes the worst case reconstruction error. Therefore, $L_2(\Sc^\text{opt})$ makes the smallest maximum angle with $PW_\omega(G)$. It is easy to show that $\cos(\theta_{\max}) = \sigma_{\min}(\Um_{\Sc\Rc})$. Thus, to find this set we need to solve a similar problem as~\eqref{eq:e-opt}. 
As stated before, this problem is combinatorial. It is possible to give a greedy algorithm to get an approximate solution. A simple greedy heuristic to approximate $\Sc^\text{opt}$ is to perform column-wise Gaussian elimination over $\Um_{\Vc\Rc}$ with partial row pivoting. The indices of the pivot rows in that case form a good estimate of $\Sc^\text{opt}$ in practice. 

Table~\ref{tab:known_spectrum results} summarizes the different set selection criteria and corresponding search algorithms under various assumptions about the signal. However, the methods described above require computation of many eigenvectors of the variation operator $\Lm$.
We circumvent this issue in the next section, by defining graph spectral proxies based on powers of $\Lm$. These spectral proxies do not require eigen-decomposition of $\Lm$ and still allow us to define a measure of quality of sampling sets. As we will show, these proxies arise naturally in the expression for the bound on the reconstruction error.
Thus, a sampling set optimal with respect to these spectral proxies ensures a small reconstruction error bound.


\section{Sampling Set Selection using Graph Spectral Proxies}
\label{sec:unknown_spectrum}
As discussed earlier, graphs considered in most real applications are very large. Hence, computing and storing the graph Fourier basis explicitly is
often 
practically infeasible. 
We now present techniques that allow us to 
express the condition for unique bandlimited reconstruction and methods for sampling set selection via simple operations using the variation operator. The following discussion holds for any choice of the variation operator $\Lm$ in Table~\ref{tab:gft_bases}.

\subsection{Cutoff Frequency}
In order to obtain a measure of quality for a sampling set $\Sc$, we first find the cutoff frequency associated with it, which can be defined as the largest frequency $\omega$ such that $\Sc$ is a uniqueness set for $PW_\omega(G)$.
It follows from Theorem \ref{thm:uniqueness} that, for $\Sc$ to be a uniqueness set of $PW_\omega(G)$, $\omega$  needs to be less than the minimum possible bandwidth that a signal in $L_2(\Sc^c)$ can have.
This would ensure that no signal from $L_2(\Sc^c)$ can be a part of $PW_\omega(G)$. Thus, the cutoff frequency $\omega_c(\Sc)$ for a sampling set $\Sc$ can be expressed as:
\begin{equation}
\omega_c(\Sc) \defeq \min_{\phi \in L_2(\Sc^c), \; \phi \neq \zerov} \omega(\phi).
\label{eq:cutoff}
\end{equation}
To use the equation above, we first need a tool to 
approximately compute the bandwidth $\omega(\phi)$ of any given signal $\phi$ without computing the Fourier coefficients explicitly.
Our proposed method for bandwidth estimation is 
based on the following definition:  
\begin{definition}[Graph Spectral Proxies]
For any signal $\fv \neq \zerov$, we define its $k^\text{th}$ spectral proxy $\omega_k(\fv)$ with $k \in \mathbb{Z}^+$ as
\begin{equation}
\omega_k(\fv) \defeq \left( \frac{\| \Lm^k \fv \|}{\|\fv\|} \right)^{1/k}.
\end{equation}
\end{definition}
\noindent For an operator $\Lm$ with real eigenvalues and eigenvectors, $\omega_k(\fv)$ can be shown to increase monotonically with $k$:
\begin{equation}
\label{eq:mom_order}
\forall \fv, k_1 < k_2 \Rightarrow \omega_{k_1}(\fv) \leq \omega_{k_2}(\fv).
\end{equation}
These quantities are bounded from above, as a result, $\lim_{k \rightarrow \infty} \omega_k(\fv)$ exists for all $\fv$. Consequently, it is easy to prove that if $\omega(\fv)$ denotes the bandwidth of a signal $\fv$, then
\begin{equation}
\label{eq:inf_order}
\forall k > 0, \; \omega_k(\fv) \leq \lim_{j\rightarrow \infty} \omega_j(\fv) = \omega(\fv). \\
\end{equation}
Note that~\eqref{eq:inf_order} also holds for an asymmetric $\Lm$ that has complex eigenvalues and eigenvectors. The proofs of \eqref{eq:mom_order} and \eqref{eq:inf_order} are provided in the Appendix.
These properties give us an important insight: as we increase the value of $k$, the spectral proxies tend to have a value close to the actual bandwidth of the signal, i.e., they essentially indicate the frequency localization of the signal energy. 
Therefore, using $\omega_k(\phi)$ as a proxy for $\omega(\phi)$ (i.e. bandwidth of $\phi$) is justified and this leads us to define the \emph{cut-off frequency estimate of order k} as
\begin{equation}
\label{eq:cutoff_def}
\Omega_k(\Sc) \defeq \min_{\phi \in L_2(\Sc^c)} \omega_k(\phi) = \min_{\phi \in L_2(\Sc^c)} \left( \frac{\|\Lm^k \phi\|}{\| \phi \|} \right)^{1/k}.
\end{equation}
Using the definitions of $\Omega_k(\Sc)$ and $\omega_c(\Sc)$ along with (\ref{eq:mom_order}) and (\ref{eq:inf_order}), we conclude that for any $k_1 < k_2$:
\begin{equation}
\omega_c(\Sc) \geq \lim_{k \rightarrow \infty} \Omega_k(\Sc) \geq \Omega_{k_2}(\Sc) \geq \Omega_{k_1}(\Sc).
\label{eq:ordering}
\end{equation}
Using (\ref{eq:ordering}) and \eqref{eq:cutoff}, we now state the following proposition:
\begin{proposition}
\label{prop:practical}
For any $k$, $\Sc$ is a uniqueness set for $PW_\omega(G)$ if, $\omega < \Omega_k(\Sc)$. $\Omega_k(\Sc)$ can be computed from (\ref{eq:cutoff_def}) as
\begin{equation}
\label{eq:compute_cutoff}
\Omega_k(\Sc) 
=  \left[ \min_{\psi}  \frac{\psi^\top ((\Lm^\top)^{k}\Lm^k)_{\Sc^c} \psi}{\psi^\top \psi} \right]^{1/2k}
= (\sigma_{1,k})^{1/2k},
\end{equation}
where $\sigma_{1,k}$
denotes the smallest eigenvalue of the reduced matrix $((\Lm^\top)^{k}\Lm^k)_{\Sc^c}$.
Further, if $\psi_{1,k}$ is the corresponding eigenvector, and $\phi_k^*$ minimizes $\omega_k(\phi)$ in (\ref{eq:cutoff_def}) (i.e. it approximates the smoothest possible signal in $L_2(\Sc^c)$), then 
\begin{equation}
\phi_k^*(\Sc^c) = \psi_{1,k}, \quad \phi_k^*(\Sc) = \zerov.
\end{equation}
\end{proposition}
\noindent We note from (\ref{eq:ordering}) that to get a better estimate of the true cut-off frequency, one simply needs a higher $k$. 
Therefore, there is a trade-off between accuracy of the estimate on the one hand, and complexity and numerical stability on the other (that arise by taking higher powers of $\Lm$). 

\subsection{Best Sampling Set of Given Size}
As shown in Proposition~\ref{prop:practical}, $\Omega_k(\Sc)$ is an estimate of the smallest bandwidth that a signal in $L_2(\Sc^c)$ can have and any signal in $PW_\omega(G)$ is uniquely sampled on $\Sc$ if $\omega < \Omega_k(\Sc)$. Intuitively, we would like the projection of $L_2(\Sc^c)$ along $PW_\omega(G)$ to be as small as possible.
Based on this intuition, we propose the following optimality criterion for selecting the best sampling set of size $m$:
\begin{equation}
\Sc^{\text{opt}}_k = \argmax_{|\Sc| = m} \Omega_k(\Sc).
\label{eq:opt_set}
\end{equation}
To motivate the above criterion more formally, let $\Pm$ denote the projector for $PW_\omega(G)$. The minimum gap~\cite{Knyazev-JFA-10} between the two subspaces $L_2(\Sc^c)$ and $PW_\omega(G)$ is given by:
\begin{align}
\inf_{\fv \in L_2(\Scc), \|\fv\|=1} \|\fv - \Pm \fv\| 
&= \sqrt{\sum_{i:\;\omega < \lambda_i} |\tilde{\fv}^*_i|^2} \nonumber \\
&\geq \sqrt{\sum_{i\in \Ic} |\tilde{\fv}^*_i|^2},
\label{eq:min_gap}
\end{align}
where $\Ic = \{i: \omega <\lambda_i \leq \Omega_k(\Sc)\}$ and $\tilde{\fv}^*_i$ denotes the $i^\text{th}$ GFT coefficient of the minimizer $\fv^*$ for the left hand side. The inequality on the right hand side holds because $\Omega_k(\Sc)$ is the smallest bandwidth that any signal in $L_2(\Scc)$ can have. Eq.~\eqref{eq:min_gap} shows that maximizing $\Omega_k(\Sc)$ increases the lower bound on the minimum gap between $L_2(\Sc^c)$ and $PW_\omega(G)$. 
The minimum gap equals $\cos(\theta_{\max})$ as defined in \eqref{eq:theta_max}~\cite{Knyazev-JFA-10}. Thus, maximizing $\Omega_k(\Sc)$ increases the lower bound on $\cos(\theta_{\max})$ which, in turn, minimizes the upper bound on the reconstruction error $\|\fv - \hat{\fv}\|$ given in \eqref{eq:error_bound}, where the original signal $\fv \notin PW_\omega(G)$ and $\hat{\fv} \in PW_\omega(G)$ is obtained by~\eqref{eq:bl_recon}.

\vspace{.25\baselineskip}
We now show that $\Omega_k(\Sc)$ also arises in the bound on the reconstruction error when the reconstruction is obtained by variational energy minimization:
\begin{equation}
\hat{\fv}_m = \argmin_{\yv \in \mathbb{R}^N} \|\Lm^m\yv\| \text{ subject to } \yv_\Sc = \fv_\Sc.
\label{eq:var_recon}
\end{equation}
It was shown in~\cite{Pesenson-CA-09} that if $\fv \in PW_\omega(G)$, then the reconstruction error $\|\hat{\fv}_m - \fv\|/\|\fv\|$, for a given $m$, is upper-bounded by $2(\omega/\Omega_1(\Sc))^m$. This bound is suboptimal and can be improved by replacing $\Omega_1(\Sc)$ with $\Omega_k(\Sc)$ (which, from~\eqref{eq:ordering}, is at least as large as $\Omega_1(\Sc)$) for any $k \leq m$, as shown in the following theorem:
\begin{theorem}
\label{thm:var_error}
Let $\hat{\fv}_m$ be the solution to \eqref{eq:var_recon} for a signal \hbox{$\fv \in PW_\omega(G)$}. Then, for any $k \leq m$,
\begin{equation}
\|\hat{\fv}_m - \fv\| \leq 2\left(\frac{\omega}{\Omega_k(\Sc)} \right)^m \|\fv\|.
\label{eq:var_err_bound}
\end{equation}
\end{theorem}
\begin{proof}
Note that $(\hat{\fv}_m - \fv) \in L_2(\Scc)$. Therefore, from~\eqref{eq:cutoff_def} 
\begin{align}
\|\hat{\fv}_m - \fv\| &\leq \frac{1}{(\Omega_m(\Sc))^m} \|\Lm^m(\hat{\fv}_m - \fv)\| \nonumber \\
\label{eq:tri_ineq}
&\leq \frac{1}{(\Omega_m(\Sc))^m} (\|\Lm^m\hat{\fv}_m\| + \|\Lm^m\fv\|) \\
\label{eq:f_hat_min}
&\leq \frac{2}{(\Omega_m(\Sc))^m} \|\Lm^m\fv\| \\
\label{eq:assume}
&\leq 2\left(\frac{\omega_m(\fv)}{\Omega_m(\Sc)} \right)^m \|\fv\| \\
&\leq 2\left(\frac{\omega}{\Omega_k(\Sc)} \right)^m \|\fv\|. \nonumber
\end{align}
\eqref{eq:tri_ineq} follows from triangle inequality. \eqref{eq:f_hat_min} holds because $\hat{\fv}_m$ minimizes $\|\Lm^m\hat{\fv}_m\|$ over all sample consistent signals. \eqref{eq:assume} follows from the definition of $\omega_m(\fv)$ and the last step follows from \eqref{eq:inf_order} and~\eqref{eq:ordering}.  
\end{proof}
%
Note that for the error bound in~\eqref{eq:var_err_bound} to go to zero as $m \to \infty$, $\omega$ must be less than $\Omega_k(\Sc)$. Thus, increasing $\Omega_k(\Sc)$ allows us to reconstruct signals in a larger bandlimited space using the variational method. Moreover, for a fixed $m$ and $k$, a higher value of $\Omega_k(\Sc)$ leads to a lower reconstruction error bound. The optimal sampling set $\Sc^{\rm opt}_k$ in \eqref{eq:opt_set} essentially minimizes this error bound.


\subsection{Finding the Best Sampling Set}
The problem posed in \eqref{eq:opt_set} is a combinatorial problem because
we need to compute $\Omega_k(\Sc)$ for every possible subset $\Sc$ of size $m$.
We therefore formulate a greedy heuristic to get an estimate of the optimal sampling set. 
Starting with an empty sampling set $\Sc$ ($\Omega_k(\Sc) = 0$) we keep adding nodes (from $\Sc^c$) one-by-one while trying to ensure maximum increase in $\Omega_k(\Sc)$ at each step.
To achieve this, we first
consider the following quantity:
\begin{align}
\label{eq:min_lambda}
\lambda^\alpha_k(\onev_\Sc) 
= \min_{\xv} \left( \omega_k(\xv) + \alpha \frac{\xv^\top{\rm diag}(\onev_{\Sc})\xv }{\xv^\top \xv} \right),
\end{align}
where $\onev_\Sc: \Vc \rightarrow \{0,1\}$ denotes the indicator function for the subset $\Sc$ (i.e. $\onev(\Sc) = \onev$ and $\onev(\Sc^c) = \zerov$). Note that the right hand side of~\eqref{eq:min_lambda}
is simply a relaxation of the constraint in~(\ref{eq:cutoff_def}). 
When $\alpha \gg 1$, the components $\xv(\Sc)$ are highly penalized during minimization, hence, forcing values of $\xv$ on $\Sc$ to be vanishingly small.
Thus, if $\xv^\alpha_k(\onev_\Sc)$ is the minimizer in (\ref{eq:min_lambda}), then $[\xv^\alpha_k(\onev_\Sc)](\Sc) \rightarrow \zerov$. 
Therefore, for $\alpha \gg 1$,
\begin{equation}
\label{eq:equiv}
\phi_k^* \approx \xv^\alpha_k(\onev_\Sc),\quad  
\Omega_k(\Sc) \approx \lambda^\alpha_k(\onev_\Sc).
\end{equation}
Now, to tackle the combinatorial nature of our problem, we allow a \emph{binary relaxation} of the indicator $\onev_S$ in~\eqref{eq:min_lambda}, to define the following quantities
\begin{align}
\omega_k^\alpha(\xv,\tv) &= \left( \omega_k(\xv) + \alpha \frac{\xv^\top \text{diag}(\tv) \xv}{\xv^\top \xv} \right), \\
\lambda_k^\alpha(\tv) &= \min_\xv \; \omega_k^\alpha(\xv,\tv),
\label{eq:lambda_relaxed}
\end{align}
where $\tv \in \mathbb{R}^N$. These relaxations circumvent the combinatorial nature of our problem and have been used recently to study graph partitioning based on Dirichlet eigenvalues~\cite{osting-arxiv-13,bourdin10}. Note that making the substitution $\tv = \onev_\Sc$ in~\eqref{eq:lambda_relaxed} gives us~\eqref{eq:min_lambda} exactly.
The effect of adding a node to $\Sc$ on $\Omega_k(\Sc)$ at each step can now be understood by observing the gradient vector $\nabla_\tv \lambda^\alpha_k(\tv)$, at $\tv = \onev_\Sc$.
Note that for any $\xv$ and $\tv$,
\begin{equation}
\frac{d \omega_k^\alpha(\xv,\tv)}{d\tv(i)} = \alpha \left( \frac{\xv(i)}{\|\xv\|} \right)^2.
\end{equation}
When $\tv = \onev_\Sc$, we know that the minimizer of \eqref{eq:lambda_relaxed} with respect to $\xv$ for large $\alpha$ is $\phi_k^*$. Hence,
\begin{equation}
\left. \frac{d \lambda_k^\alpha(\tv)}{d\tv(i)} \right|_{\tv = \onev_\Sc} = \left. \frac{d \omega_k^\alpha(\phi_k^*,\tv)}{d\tv(i)} \right|_{\tv = \onev_\Sc} = \alpha \left( \frac{\phi_k^*(i)}{\|\phi_k^*\|} \right)^2.
\end{equation}
The equation above gives us the desired greedy heuristic - starting with an empty $\Sc$ (i.e., $\onev_\Sc = \zerov$), if at each step, we include the node on which the smoothest signal $\phi_k^* \in L_2(\Sc^c)$ has maximum energy (i.e. $\onev_\Sc(i) \leftarrow 1, i = \text{arg max}_j(\phi_k^*(j))^2$), then $\lambda^\alpha_k(\tv)$ and in effect, the cut-off estimate $\Omega_k(\Sc)$, tend to increase maximally. We summarize the method for estimating $\Sc^{\text{opt}}_k$ in Algorithm \ref{alg:alg1}.

One can show that the cutoff frequency estimate $\Omega_k(\Sc)$  associated with a sampling set can only increase (or remain unchanged) when a node is added to it. This is stated more formally in the following proposition.
\begin{proposition}
Let $\Sc_1$ and $\Sc_2$ be two subsets of nodes of $G$ with $\Sc_1 \subseteq \Sc_2$. Then $\Omega_k(\Sc_1) \leq \Omega_k(\Sc_2)$.
\end{proposition}
This turns out to be a straightforward consequence of the eigenvalue interlacing property for symmetric matrices. 
\begin{theorem}[Eigenvalue interlacing~\cite{Haemers-LAA-95}] Let $\Bm$ be a symmetric $n \times n$ matrix. Let $\Rc = \{1,2,\ldots, r\}$, for $1 \leq r \leq n-1$ and $\Bm_r = \Bm_{\Rc}$. Let $\lambda_k(\Bm_r)$ be the $k$-th largest eigenvalue of $\Bm_r$. Then the following interlacing property holds:
\begin{align*}
\lambda_{r+1}(\Bm_{r+1}) \leq \lambda_r(\Bm_r) \leq &\lambda_r(\Bm_{r+1}) \leq \ldots\\
\ldots \leq &\lambda_2(\Bm_{r+1}) \leq \lambda_1(\Bm_r) \leq \lambda_1(\Bm_{r+1}).
\end{align*}
\end{theorem}
The above theorem implies that if $\Sc_1 \subseteq \Sc_2$, then $\Sc_2^c \subseteq \Sc_1^c$ and thus, $\lambda_{\min}\left[\left((\Lm^\top)^{k}\Lm^k \right)_{\Sc^c_1}\right] \leq \lambda_{\min}\left[\left((\Lm^\top)^{k}\Lm^k \right)_{\Sc^c_2}\right]$. 
\begin{algorithm}[ht]
\caption{Greedy heuristic for estimating $\Sc^{\text{opt}}_k$}
\label{alg:alg1}
\begin{algorithmic}[1]
\REQUIRE $G = \{\Vc,E\}$, $\Lm$, sampling set size $M$, $k \in \mathbb{Z}^+$.
\ENSURE $\Sc = \{ \emptyset \}$.
\WHILE{$|\Sc|<M$} 
\STATE For $\Sc$, compute smoothest signal $\phi_k^* \in L_2(\Sc^c)$ using Proposition \ref{prop:practical}.
\STATE $v \leftarrow \text{arg max}_i (\phi_k^*(i))^2$.
\STATE $\Sc \leftarrow \Sc \cup v$.
\ENDWHILE
\STATE $\Sc^\text{opt}_k \leftarrow \Sc$.
 \end{algorithmic}
\end{algorithm}
\vspace{0.5\baselineskip}
\subsubsection*{Connection with Gaussian elimination}
From Section~\ref{sec:known_spectrum}, we know that the optimal sampling set can be obtained by maximizing $\sigma_\text{min} \left(\Um_{\Sc\Rc}\right)$ with respect to $\Sc$. A heuristic to obtain good sampling sets is to perform a column-wise Gaussian elimination with pivoting on the eigenvector matrix $\Um$. Then, a sampling set of size $i$ is given by the indices of zeros in the $(i+1)^\text{th}$ column of the echelon form. 
We now show that the greedy heuristic proposed in Algorithm \ref{alg:alg1} is closely related to this rank-revealing Gaussian elimination procedure through the following observation: 
\begin{proposition}
Let $\Phi$ be the matrix whose columns are given by the smoothest signals $\phi_\infty^*$ obtained sequentially after each iteration of Algorithm \ref{alg:alg1} with $k = \infty$, (i.e., $\Phi = \left[\phi^*_\infty|_{|\Sc| = 0} \; \phi^*_\infty|_{|\Sc| = 1}, \; \dots \right]$). Further, let $\Tm$ be the matrix obtained by performing column-wise Gaussian elimination on $\Um$ with partial pivoting. Then, the columns of $\Tm$ are equal to the columns of $\Phi_\infty^*$ within a scaling factor.
\end{proposition}
\begin{proof}
If $\Sc$ is the smallest sampling set for uniquely representing signals in $PW_\omega(G)$ and $r = {\rm dim}\; PW_\omega(G)$, then we have the following:
\begin{enumerate}[leftmargin=15pt]
\item $|\Sc| = r$.
\item The smoothest signal $\phi^*_\infty \in L_2(\Sc^c)$ has bandwidth $\lambda_{r+1}$.
\end{enumerate}
Therefore, $\phi^*_\infty|_{|\Sc|=r}$ is spanned by the first $r+1$ frequency basis elements $\{ \uv_1, \dots, \uv_{r+1}\}$. Further, since $\phi^*_\infty|_{|\Sc|=r}$ has zeroes on exactly $r$ locations, it can be obtained by performing Gaussian elimination on $\uv_{r+1}$ using $\uv_1, \uv_2, \dots, \uv_{r}$. Hence the $(r+1)^\text{th}$ column of $\Phi$ is equal (within a scaling factor) to the $(r+1)^\text{th}$ column of $\Tm$. Pivoting comes from the fact that the $(i+1)^\text{th}$ sampled node is given by the index of the element with maximum magnitude in $\phi_\infty^*|_{|\Sc|=i}$, and is used as the pivot to zeros out elements with same index in subsequent columns. 
\end{proof}
The above result illustrates that Algorithm \ref{alg:alg1} is an iterative procedure that approximates a rank-revealing Gaussian elimination procedure on $\Um_{\Vc\Rc}$. For the known-spectrum case, this is a good heuristic for maximizing $\sigma_{min} \left( \Um_{\Sc\Rc} \right)$. In other words, our method directly maximizes $\sigma_{min}\left( \Um_{\Sc\Rc} \right)$ without going through the intermediate step of computing $\Um_{\Vc\Rc}$. As we shall see in the next subsection, this results in significant savings in both time and space complexity.


\subsection{Complexity and implementation issues}
\begin{table*}[t]
\centering
\caption{Comparison of complexity of different sampling set selection algorithms.}
\label{tab:complexity}
\begin{tabular}{llll}
\hline \hline
                           & Method in~\cite{Shomorony-GSIP-14} & Method in~\cite{Chen-arxiv-15} & Proposed Method \\ \hline
Eigen-pair computations & $O((|E||\Sc| + C|\Sc|^3)T_1)$                         & $O((|E||\Sc| + C|\Sc|^3)T_1)$                     & $O\left(k|E||S|T_2(k)\right)$    \\ \hline
Sampling set search        & $O(N|\Sc|^3)$                        & $O(N|\Sc|^4)$                    & $O(N|\Sc|)$     \\ \hline
Space complexity           & $O(N|\Sc|)$                              & $O(N|\Sc|)$                          & $O(N)$       \\ \hline  
\end{tabular}
\end{table*}
We note that in the algorithm, computing the first eigen-pair of $((\Lm^\top)^{k}\Lm^k)_{\Sc^c}$ is the major step for each iteration. 
There are many efficient iterative methods, such as those based on Rayleigh quotient minimization, for computing the smallest eigen-pair of a matrix \cite{Knyazev-SIAM-01}. The atomic step in all of these methods consists of matrix-vector products. Specifically, in our case, this step involves evaluating the expression $((\Lm^\top)^{k}\Lm^k)_{\Sc^c}\xv$. Note that we do not actually need to compute the matrix $((\Lm^\top)^{k}\Lm^k)_{\Sc^c}$ explicitly, since the expression can be implemented as a sequence of matrix-vector products as 
\begin{equation}
((\Lm^\top)^{k}\Lm^k)_{\Sc^c}\xv = \Id_{\Scc\Vc} \Lm^\top \ldots \Lm^\top \Lm \ldots \Lm \Id_{\Vc\Scc} \xv.
\end{equation}
Evaluating the expression involves $2k$ matrix-vector products and has a complexity of $O(k|E|)$, where $|E|$ is the number of edges in the graph. Moreover, a localized and parallel implementation of this step is possible in the case of sparse graphs. The number of iterations required for convergence of the eigen-pair computation iterations is a function of the eigen-value gaps \cite{Knyazev-SIAM-01} and hence dependent on the graph structure and edge-weights.

For the methods of ~\cite{Shomorony-GSIP-14} and~\cite{Chen-arxiv-15}, one needs to compute a portion of the eigenvector matrix, i.e., $\Um_{\Vc\Sc}$ (assuming $|\Rc| = |\Sc|$). This can be done using block-based Rayleigh quotient minimization methods~\cite{Knyazev-SIAM-01}, block-based Kryolov subspace methods such as Arnoldi/Lanczos iterations or deflation methods in conjunction with single eigen-pair solvers~\cite{saad11}. The complexity of these methods increases considerably as the number of requested eigen-pairs increases, making them impractical. On the other hand, our method requires computing a single eigen-pair at each iteration, making it viable for cases when a large number of samples are required.
Moreover, the sample search steps in the methods of ~\cite{Shomorony-GSIP-14} and~\cite{Chen-arxiv-15} require an SVD solver and a linear system solver, respectively, thus making them much more complex in comparison to our method, where we only require finding the maximum element of a vector. 
Our algorithm is also efficient in terms of space complexity, since at any point we just need to store $\Lm$ and one vector. On the other hand, ~\cite{Shomorony-GSIP-14, Chen-arxiv-15} require storage of at least $|\Sc|$ eigenvectors. 

A summary of the complexities of all the methods is given in Table~\ref{tab:complexity}. The eigen-pair computations for  $\cite{Shomorony-GSIP-14, Chen-arxiv-15}$ are assumed to be performed using a block version of the Rayleigh quotient minimization method, which has a complexity of $O((|E||\Sc| + C|\Sc|^3)T_1)$, where $T_1$ denotes the number of iterations for convergence, and $C$ is a constant. The complexity of computing one eigen-pair in our method is $O(k|E||S|T_2(k))$, where $T_2(k)$ denotes the average number of iterations required for convergence of a single eigen-pair. $T_1$ and $T_2(k)$ required to achieve a desired error tolerance are functions of the eigen-gaps of $\Lm$ and $\Lm^k$ respectively.
In general, $T_2(k) > T_1$, since $\Lm^k$ has lower eigengaps near the smallest eigenvalue. 
Increasing the parameter $k$ further flattens the spectrum of $\Lm^k$ near the smallest eigenvalue leading to an increase in $T_2(k)$, since one has to solve a more ill-conditioned problem.
We illustrate this in the next section through experiments that compare the running times of all the methods.

The choice of the parameter $k$ depends on the desired accuracy -- a larger value of $k$ gives a better sampling set, but increases the complexity proportionally, thus providing a trade-off.  Through experiments, we show in the next section that the quality of the sampling set is more sensitive to choice of $k$ for sparser graphs. This is because increasing $k$ results in the consideration of more \emph{global} information while selecting samples. On the other hand, dense graphs have a lower diameter and there is relatively little information to be gained by increasing $k$.




\section{Experiments}
\label{sec:experiments}
We now numerically evaluate the performance of the proposed work. 
The experiments involve comparing the reconstruction errors and running times of different sampling set selection algorithms in conjunction with consistent bandlimited reconstruction~\eqref{eq:bl_recon}\footnote{Although reconstruction using~\eqref{eq:bl_recon} requires explicit computation of $\Um_{\Vc\Rc}$, there exist efficient localized reconstruction algorithms that circumvent this~\cite{Narang-GlobalSIP-13,wang-tsp-15}. However, in the current work, we restrict our attention to the problem of sampling set selection.}. We compare our approach with the following methods:
\begin{enumerate}
\item [M1:] 
This method ~\cite{Chen-arxiv-15} uses a greedy algorithm to approximate the $\Sc$ that maximizes $\sigma_{\min} (\Um_{\Sc\Rc})$. Consistent bandlimited reconstruction~\eqref{eq:bl_recon} is then used to estimate the unknown samples.
\item [M2:] 
At each iteration $i$, this method~\cite{Shomorony-GSIP-14} finds the representation of $\uv_i$ as $\sum_{j<i} \beta_j \uv_j + \sum_{u \notin \Sc} \alpha_u \onev_u$, where $\onev_u$ is the delta function on $u$. The node $v$ with maximum $|\alpha_v|$ is sampled. Reconstruction is done using~\eqref{eq:bl_recon}.
\end{enumerate}
Both the above methods assume that a portion of the frequency basis is known and the signal to be recovered is exactly bandlimited.
As a baseline, we also compare all sampling set selection methods against uniform random sampling.


\subsection{Examples with Artificial Data}
We first give some simple examples on the following simulated undirected graphs:
\begin{enumerate}
\item [$G1$:] Erd{\"o}s-Renyi random graph (unweighted) with $1000$ nodes and connection probability $0.01$.
\item [$G2$:] Small world graph~\cite{Watts-nature} (unweighted) with $1000$ nodes. The underlying regular graph with degree $8$ is rewired with probability $0.1$.
\item [$G3$:] Barab{\'a}si-Albert random network~\cite{barabasi-science-99} with $1000$ nodes. The seed network is a fully connected graph with $m_0 = 4$ vertices, and each new vertex is connected to $m = 4$ existing vertices randomly. This model, as opposed to $G1$ and $G2$, is a scale-free network, \emph{i.e.}, its degrees follow a power law $P(k) \sim k^{-3}$.
\end{enumerate}
The performance of the sampling methods depends on the assumptions about the true signal and sampling noise. For each of the above graphs, we consider the problem in the following scenarios: 
\begin{enumerate}
\item [$F1$:] The true signal is noise-free and exactly bandlimited with $r = \dim{PW_\omega(G)} = 50$. The non-zero GFT coefficients are randomly generated from $\Nc(1,0.5^2)$.
\item [$F2$:] The true signal is exactly bandlimited with $r = 50$ and non-zero GFT coefficients are generated from $\Nc(1,0.5^2)$. The samples are noisy with additive iid Gaussian noise such that the SNR equals 20dB.
\item [$F3$:] The true signal is approximately bandlimited with an exponentially decaying spectrum. Specifically, the GFT coefficients are generated from $\Nc(1,0.5^2)$, followed by rescaling with the following filter (where $r = 50$):
\begin{equation}
h(\lambda) = 
\begin{cases}
1,  &\lambda < \lambda_r \\
e^{-4(\lambda - \lambda_r)}, & \lambda \geq \lambda_r.
\end{cases}
\end{equation}

\end{enumerate}

We generate 50 signals from each of the three signal models on each of the graphs, use the sampling sets obtained from the all the methods to perform reconstruction and plot the mean of the mean squared error (MSE) for different sizes of sampling sets. For our algorithm, we set the value of $k$ to 2, 8 and 14. The result is illustrated in Figure~\ref{fig:toy_results}. Note that when the size of the sampling set is less than $r = 50$, the results are quite unstable. This is expected, because the uniqueness condition is not satisfied by the sampling set. Beyond $|\Sc| = r$, we make the following observations:

\begin{enumerate}[leftmargin=15pt]
\item For the noise-free, bandlimited signal model F1, all methods lead to zero reconstruction error as soon as the size of the sampling set exceeds the signal cutoff $r = 50$ (error plots for this signal model are not shown). This is expected from the sampling theorem. 
%
It is interesting to note that in most cases, uniform random sampling does equally well, since the signal is noise-free and perfectly bandlimited. 
\item For the noisy signal model F2 and the approximately bandlimited model F3, our method has better or comparable performance in most cases. This indicates that our method is fairly robust to noise and model mismatch. Uniform random sampling performs very badly as expected, because of lack of stability considerations.
\end{enumerate}

\begin{figure*}
        \centering
                \begin{subfigure}[b]{0.32\textwidth}
                \includegraphics[width=\textwidth]{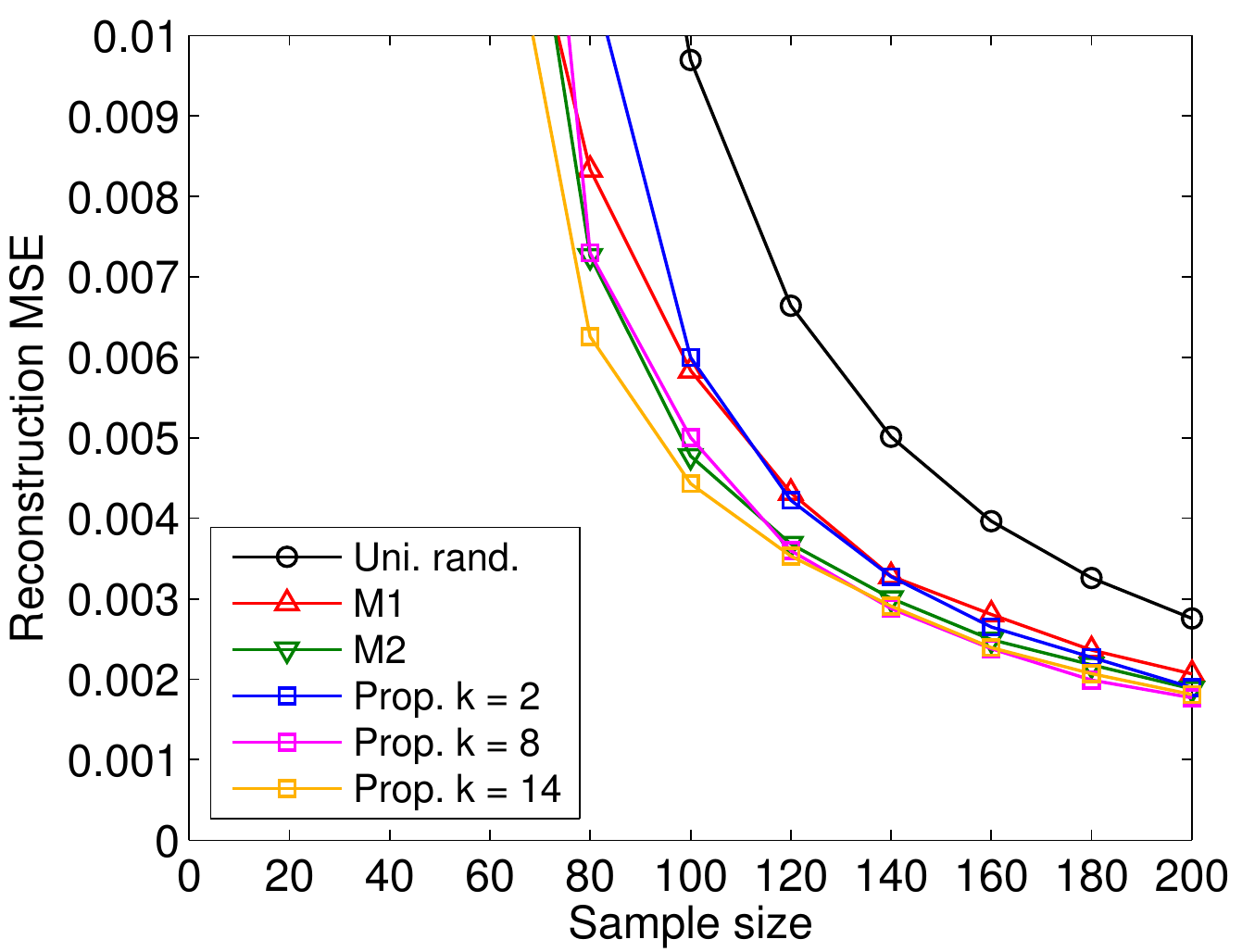}
                \caption{Graph G1 and signal model F2}
                \label{fig:g1_f2}
        \end{subfigure}%
        ~
        \begin{subfigure}[b]{0.32\textwidth}
                \includegraphics[width=\textwidth]{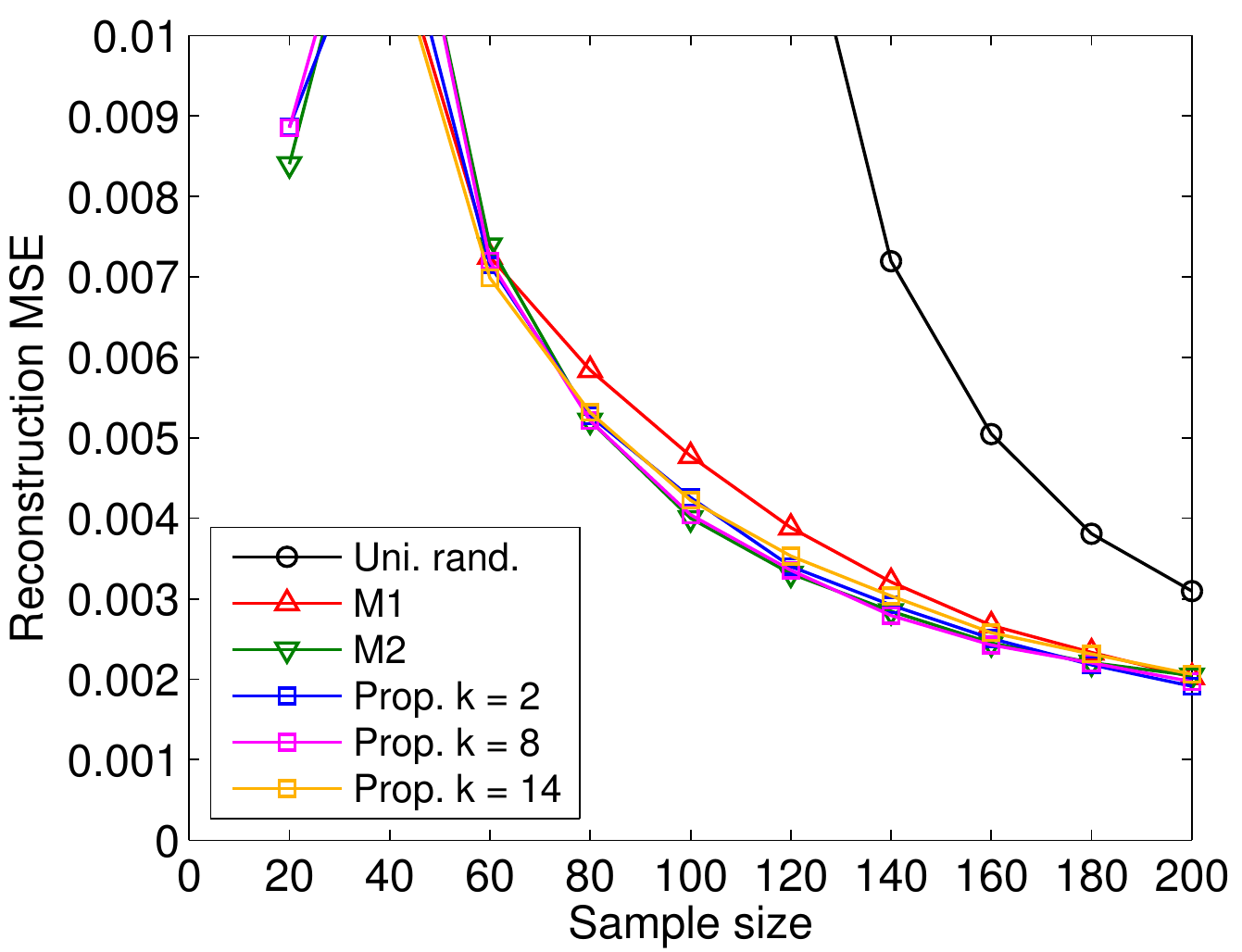}
                \caption{Graph G2 and signal model F2}
                \label{fig:g2_f2}
        \end{subfigure}%
        ~
        \begin{subfigure}[b]{0.32\textwidth}
                \includegraphics[width=\textwidth]{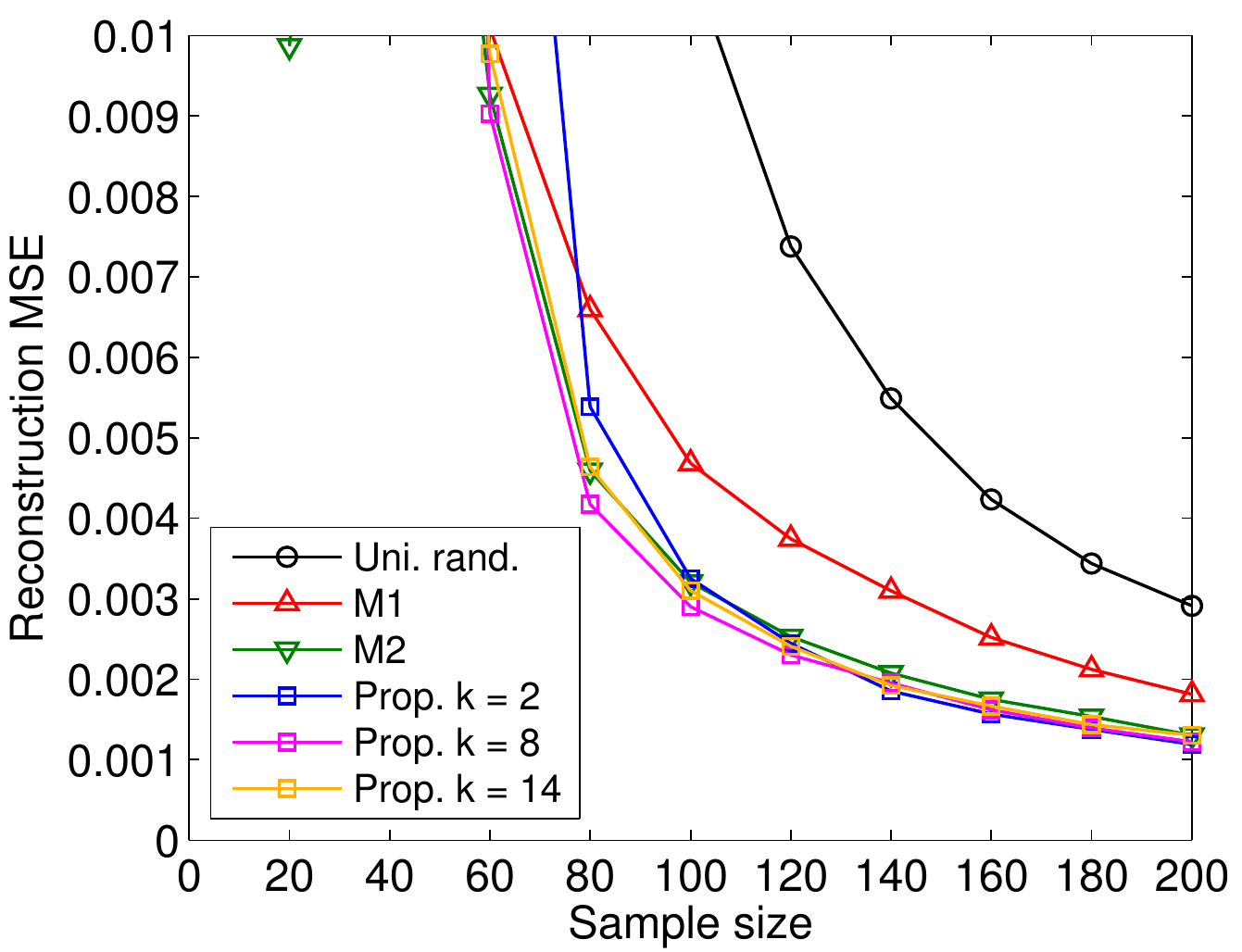}
                \caption{Graph G3 and signal model F2}
                \label{fig:g3_f2}
        \end{subfigure}%
        \\
                \begin{subfigure}[b]{0.32\textwidth}
                \includegraphics[width=\textwidth]{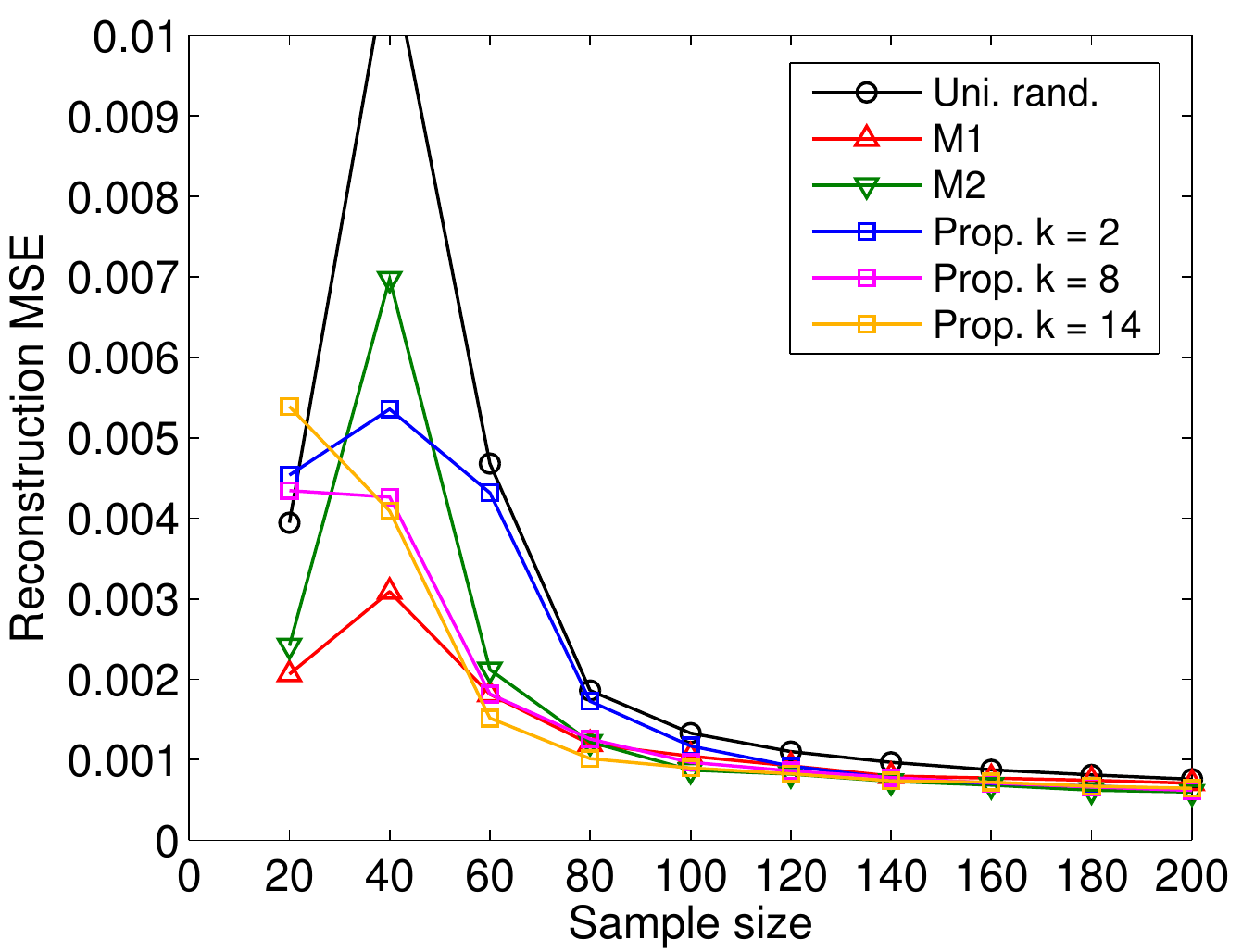}
                \caption{Graph G1 and signal model F3}
                \label{fig:g1_f3}
        \end{subfigure}%
        ~
        \begin{subfigure}[b]{0.32\textwidth}
                \includegraphics[width=\textwidth]{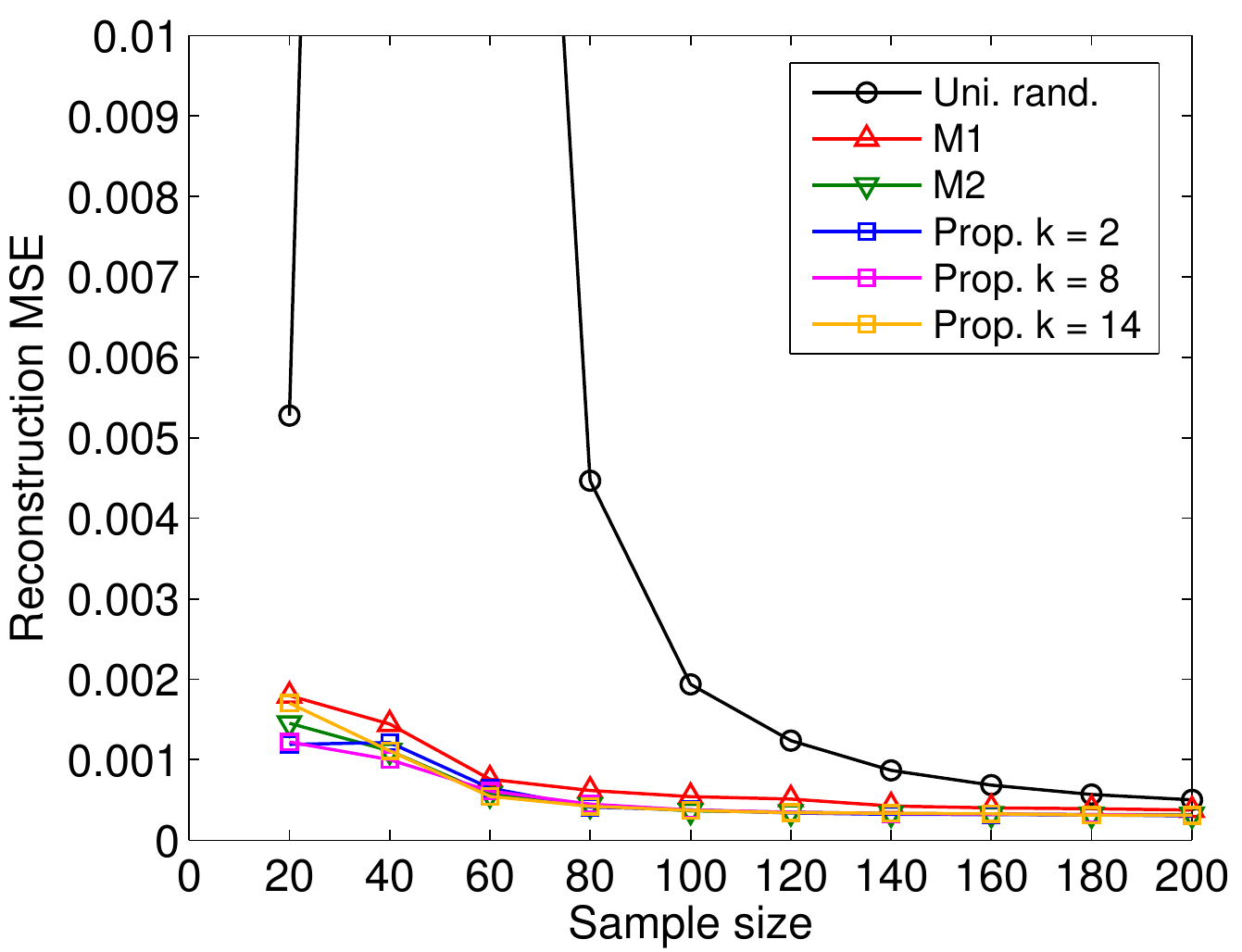}
                \caption{Graph G2 and signal model F3}
                \label{fig:g2_f3}
        \end{subfigure}%
        ~
        \begin{subfigure}[b]{0.32\textwidth}
                \includegraphics[width=\textwidth]{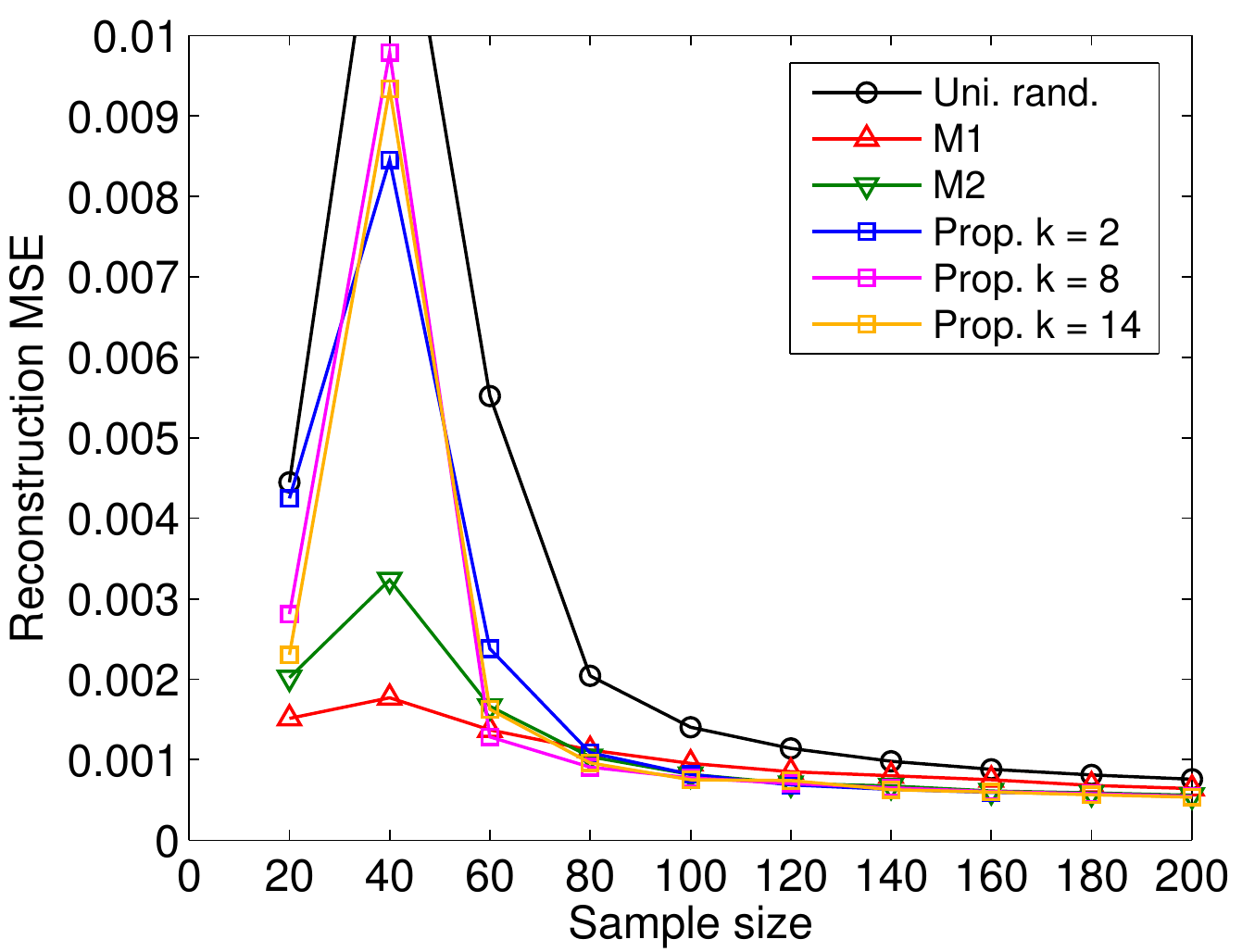}
                \caption{Graph G3 and signal model F3}
                \label{fig:g3_f3}
        \end{subfigure}%
        \caption{Reconstruction results for different graph and signal models. Plots for signal model F1 are not shown since the reconstruction errors are identically zero for all methods when $|\Sc| \geq \dim PW_\omega(G) = 50$. The large reconstruction errors for $|\Sc|< 50$ arise due to non-uniqueness of bandlimited reconstruction and hence, are less meaningful.}
        \label{fig:toy_results}
\end{figure*}


\subsubsection*{Effect of parameter $k$ in the spectral proxy}
Parameter $k$ in the definition of spectral proxies controls how closely we estimate the bandwidth of any signal $\fv$. Spectral proxies with higher values of $k$ give a better approximation of the bandwidth. Our sampling set selection algorithm tries to 
maximize the smallest bandwidth that a signal in $L_2(\Scc)$ can have. Using higher values of $k$ allows us to estimate this smallest bandwidth more closely, thereby leading to better 
sampling sets as demonstrated in Figure~\ref{fig:effect_of_k}. Intuitively, maximizing $\Omega_k(\Sc)$ with $k=1$ ensures that the sampled nodes are well connected to the unsampled nodes~\cite{Gadde-KDD-14} and thus, allows better propagation of the observed signal information. Using $k>1$ takes into account multi-hop paths while ensuring better connectedness between $\Sc$ and $\Scc$. This effect is especially important in sparsely connected graphs and the benefit of increasing $k$ becomes less noticeable when the graphs are dense as seen in Figure~\ref{fig:effect_of_k}. However, this improvement in performance in the case of sparse graphs comes at the cost of increased numerical complexity. 
\begin{figure*}
\centering
\begin{subfigure}[b]{0.32\textwidth}
\includegraphics[width=\textwidth]{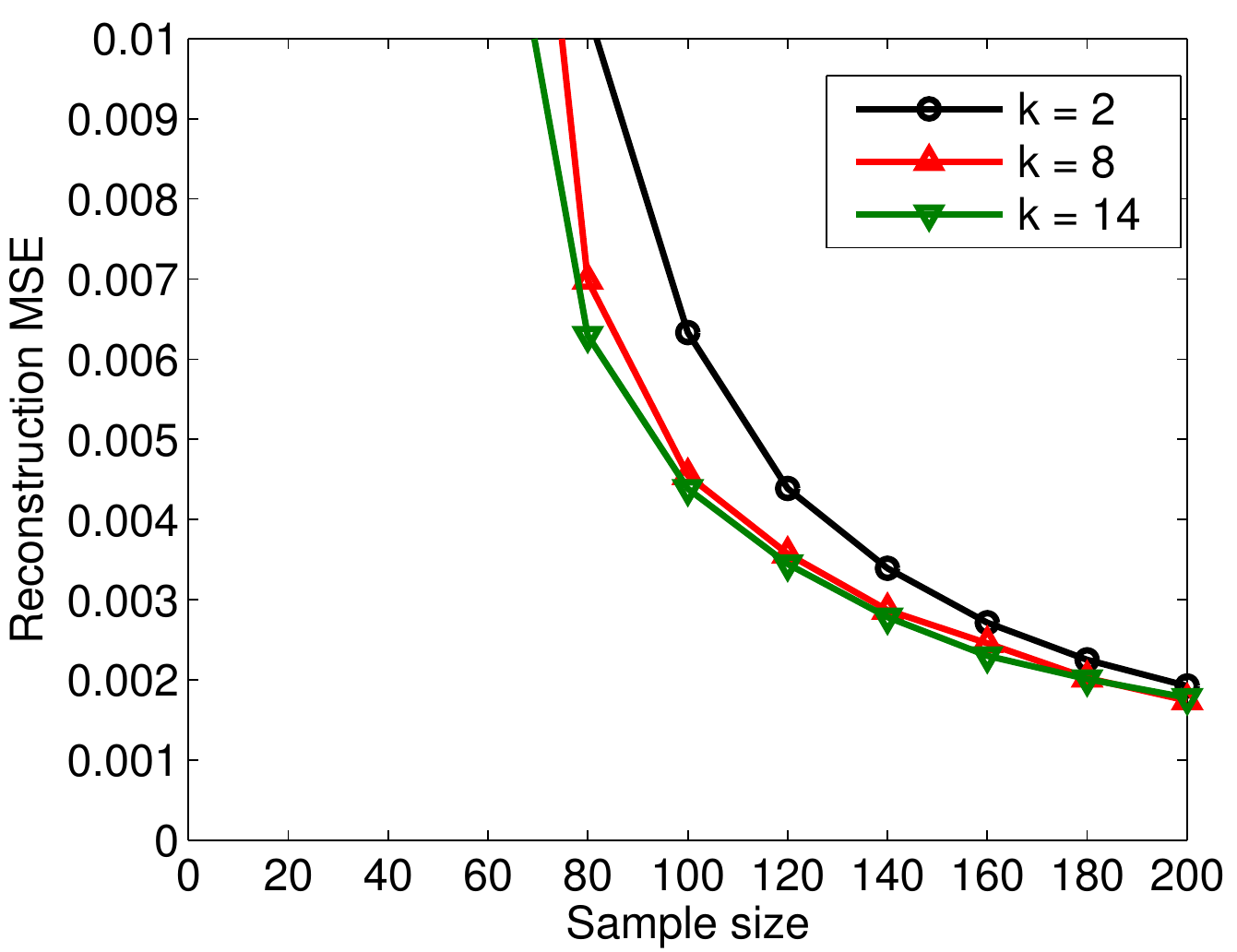}
\caption{$p = 0.01$}
\end{subfigure}
~
\begin{subfigure}[b]{0.32\textwidth}
\includegraphics[width=\textwidth]{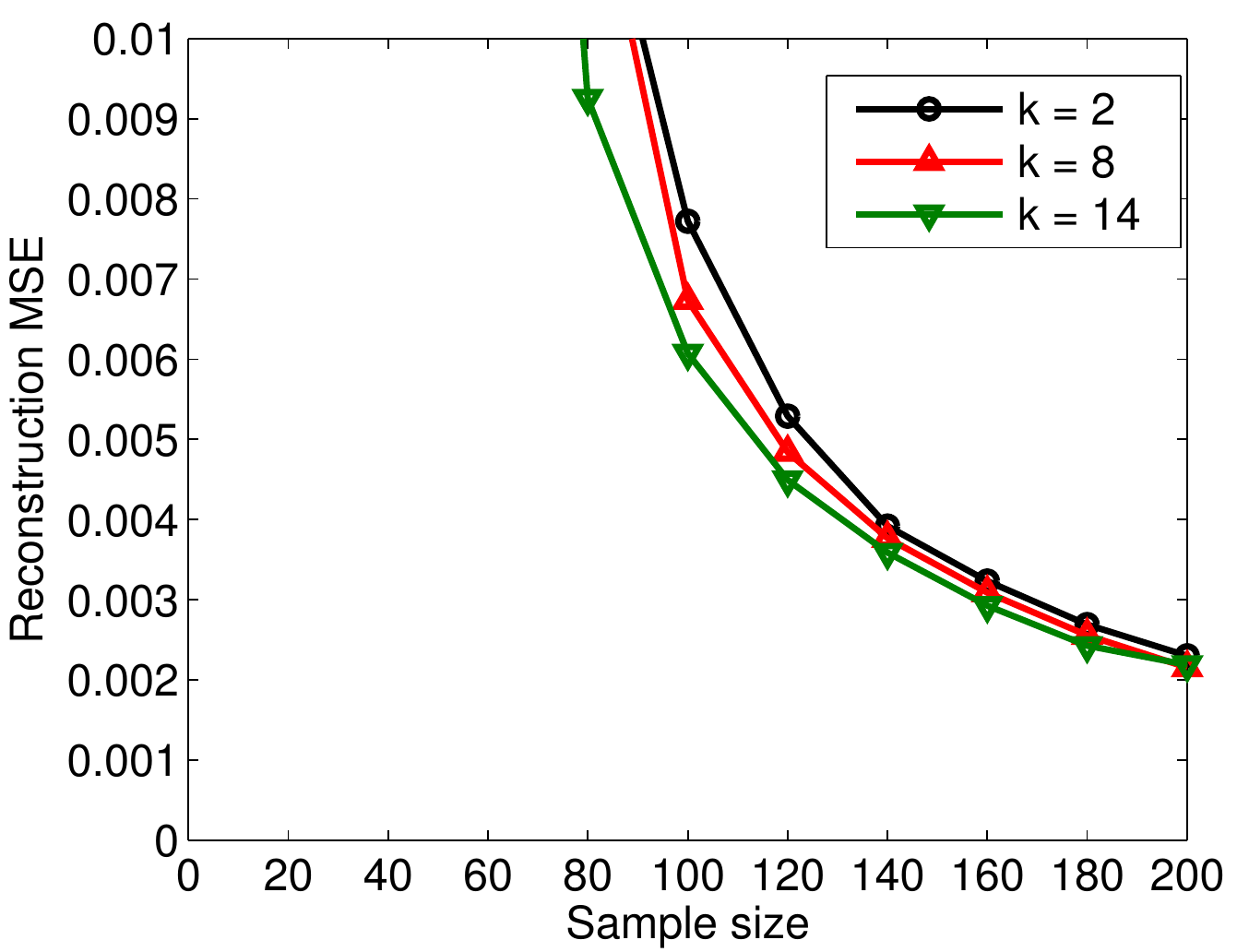}
\caption{$p = 0.05$}
\end{subfigure}
~
\begin{subfigure}[b]{0.32\textwidth}
\includegraphics[width=\textwidth]{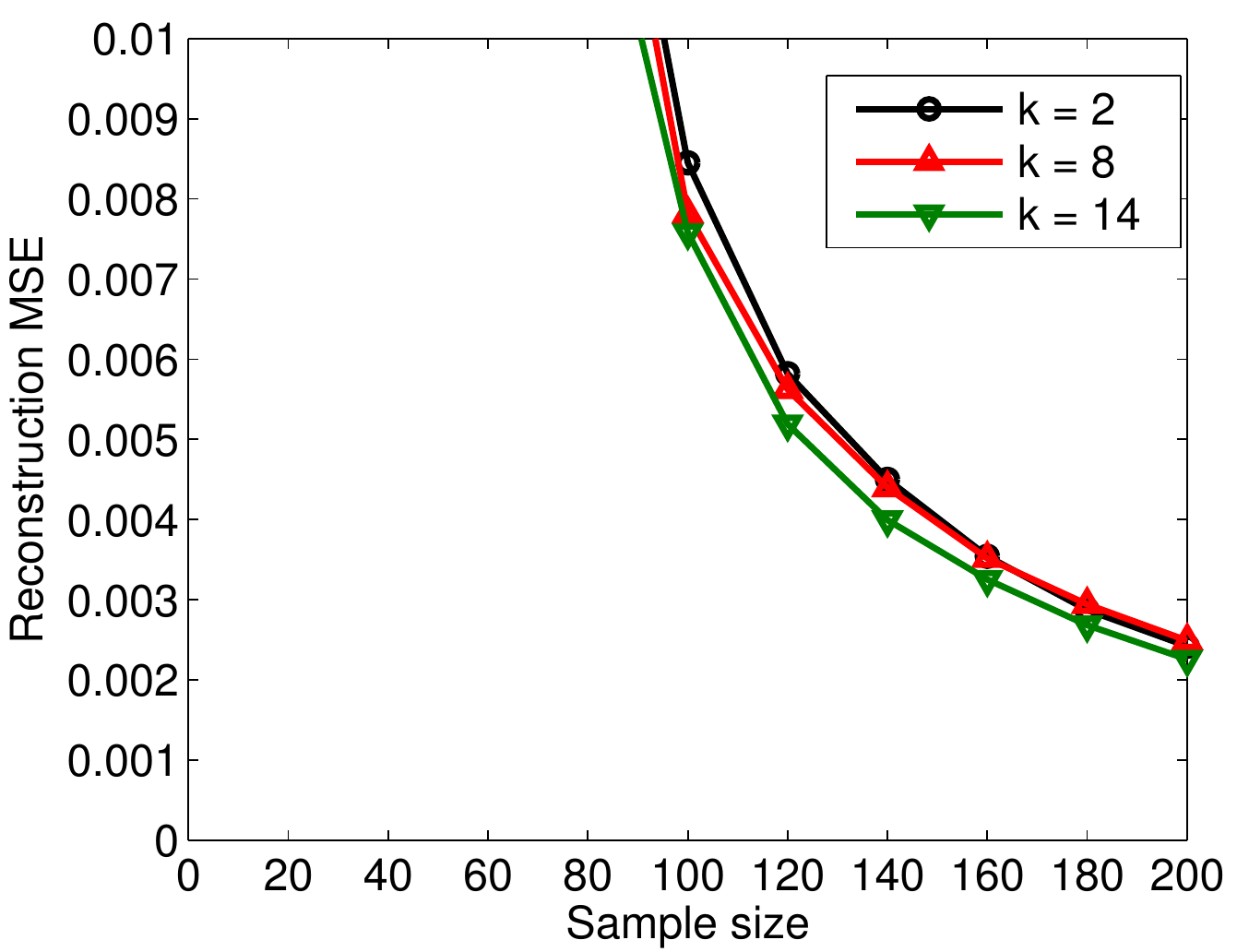}
\caption{$p = 0.1$}
\end{subfigure}
\caption{Reconstruction performance for noisy signals (model F2) with different values of $k$ in Erd\"{o}s-Renyi graphs having different connection sparsity levels. 
Higher connection probability $p$ implies lower sparsity.}
\label{fig:effect_of_k}
\end{figure*}

\subsubsection*{Running time}
We also compare the running times of the sampling set selection methods for different sizes of the graph. For our experiments, we generate symmetrized Erd{\"o}s-Renyi random graphs of different sizes with parameter $0.01$, and measure the average running time of selecting $5\%$ of the samples in MATLAB. For computing the eigen-pairs, we use the code for the Locally Optimal Block Prec-conditioned Conjugate Gradient (LOBPCG) method available online~\cite{Knyazev-SIAM-01} (this was observed to be faster than MATLAB's inbuilt sparse eigensolver \texttt{eigs}, which is based on Lanczos iterations). The results of the experiments are shown in Table~\ref{tab:running_time}. We observe that the rate of increase of running time as the graph size increases is slower for our method compared to other methods, thus making it more practical. Note that the increase with respect to $k$ is nonlinear since the eigengaps are a function of $k$ and lead to different number of iterations required for convergence of the eigenvectors.

\begin{table}[ht]
\centering
\caption{Running time of different methods (in seconds) for selecting $5\%$ samples on graphs of different sizes. The running time for $M1$ increases drastically and is ignored beyond graph size 5k.}
\label{tab:running_time}
\begin{tabular}{l|r|r|r|r}
\hline \hline
                         & 1k & 5k  & 10k & 20k \\
                         \hline
M1                       & $16.76$  & $12,322.72$ & -           & -           \\
\hline
M2                       & $2.16$   & $57.46$     & $425.92$    & $3004.01$   \\
\hline
Proposed, $k = 4$ & $2.00$   & $11.13$     & $84.85$     & $566.39$    \\
\hline
Proposed, $k = 6$ & $13.08$  & $24.46$     & $170.15$    & $1034.21$   \\
\hline
Proposed, $k = 8$ & $31.16$  & $53.42$     & $316.12$    & $1778.31$  \\
\hline
\end{tabular}
\end{table}

\subsection{A Real World Example}
In this example, we apply the proposed method to classification of the USPS handwritten digit dataset~\cite{usps-data}. This dataset consists of $1100$ images of size $16 \times 16$ each corresponding digits 0 to 9. 
We create $10$ different subsets of this dataset randomly, consisting of $100$ images from each class.
The data points can be thought of as points $\{\xv_i\}_{i = 1}^{1000} \subset \mathbb{R}^{256}$ with labels $\{\yv_i\}_{i = 1}^{1000}$. 
For each instance, we construct a symmetrized $k$-nearest neighbor ($k$-nn) graph with $k = 10$, where each node corresponds to a data point.
We restrict the problem to the largest strongly connected component of the graph for convenience so that a stationary distribution for the resultant random walk exists which allows us to define the random walk based GFT. 
The graph signal is given by the membership function $\fv^c$ of each class $c$ which takes a value $1$ on a node which belongs to the class and is $0$ otherwise. To solve the multi-class classification task, we use the one-vs-rest strategy which entails reconstructing the membership function of every class. The final classification for node $i$ is then obtained by
\begin{equation}
\yv_i = \argmax_c \; \{\fv^c_i\}.
\end{equation}

\begin{figure*}
        \centering
        \begin{subfigure}[b]{0.32\textwidth}
                \includegraphics[width=\textwidth]{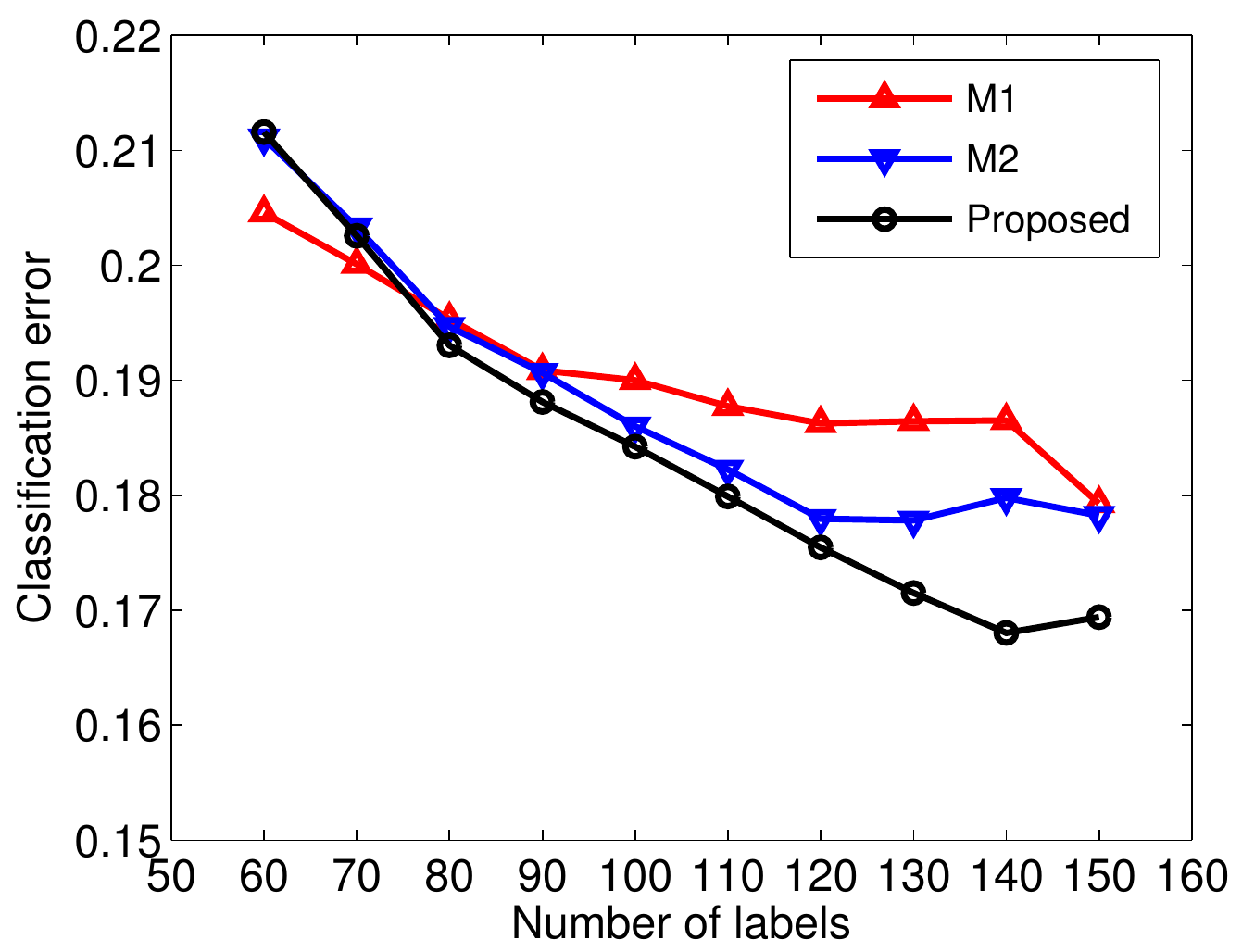}
                \caption{Comparison of different methods using the adjacency based variation operator.}
                \label{fig:usps_methods}
        \end{subfigure}%
        ~\;
        \begin{subfigure}[b]{0.32\textwidth}
                \includegraphics[width=\textwidth]{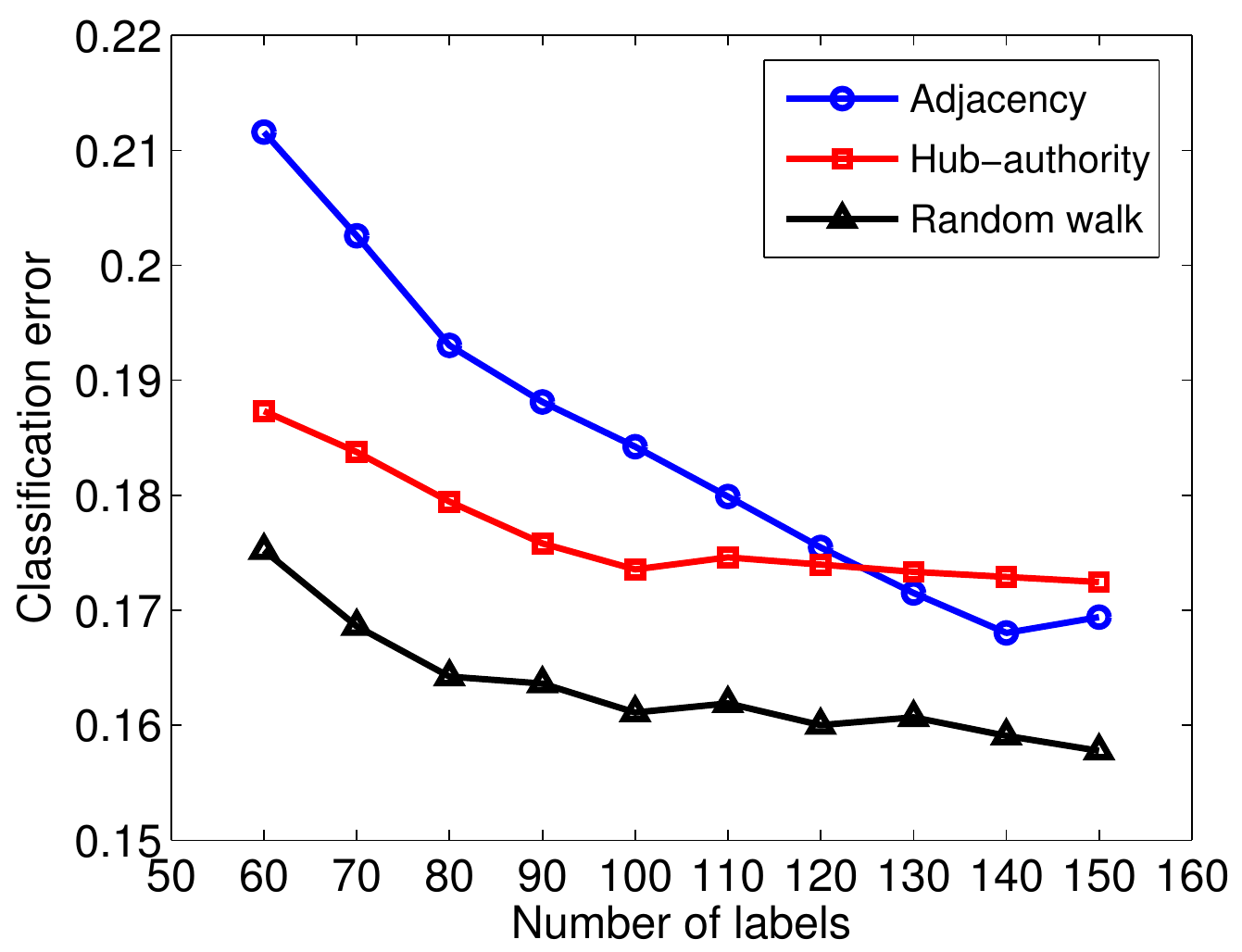}
                \caption{Performance of proposed method with different choices of the variation operator.}
                \label{fig:usps_bases}
        \end{subfigure}%
        ~\;
        \begin{subfigure}[b]{0.32\textwidth}
                \includegraphics[width=\textwidth]{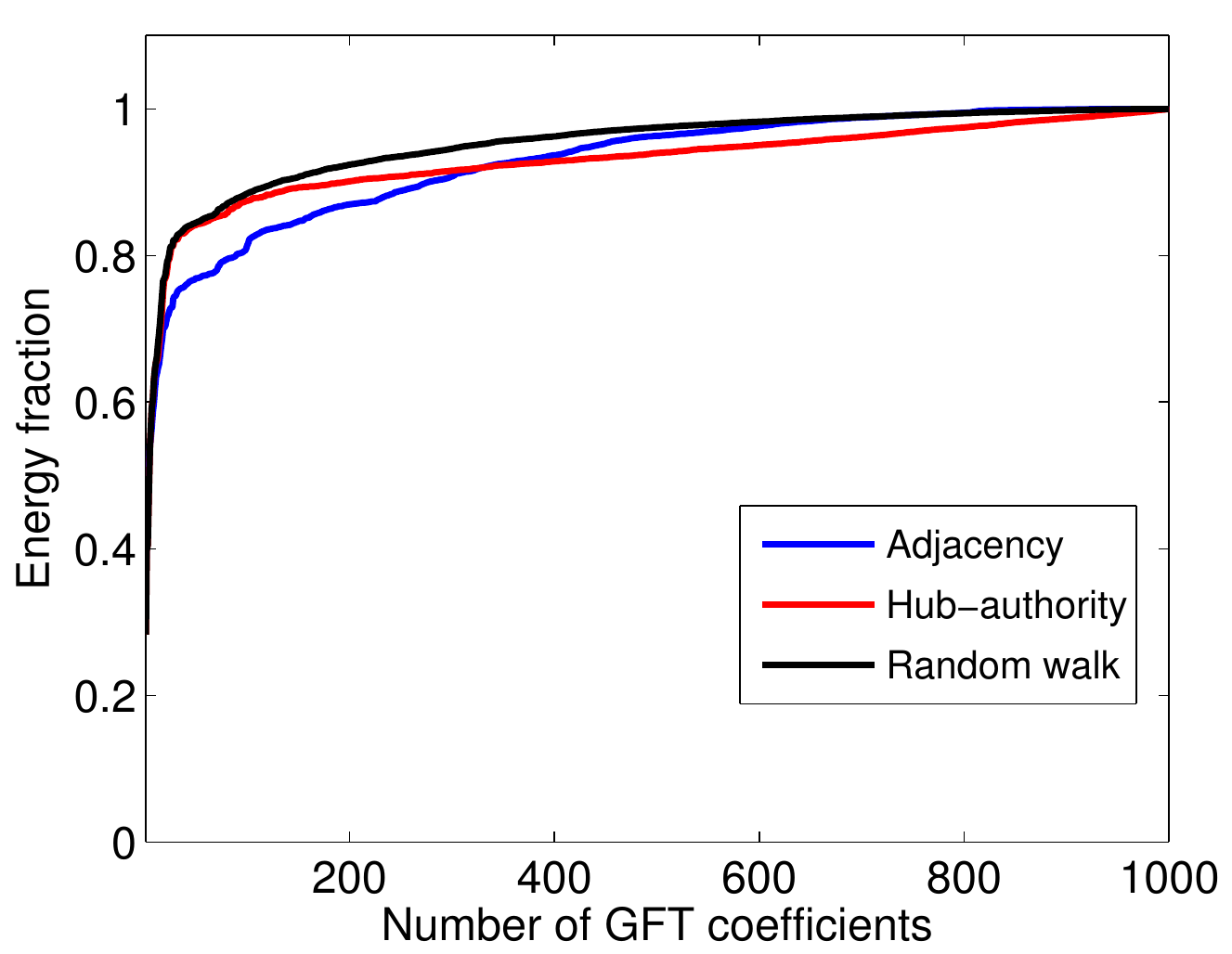}
                \caption{Energy compaction with different GFTs obtained from different variation operators.}
                \label{fig:usps_gft}
        \end{subfigure}%
        \caption{Classification results for the USPS dataset using different methods and GFTs}
        \label{fig:usps}
\end{figure*}

We first compare the performance of the proposed method against M1 and M2 using the normalized adjacency matrix based GFT with the variation operator $\Lm = \mathbf{I} - \Dm^{-1}\Wm$. The bandwidth parameter $r$ is set to $50$.
The plot of classification error averaged over the 10 dataset instances vs. number of labels is presented in Figure~\ref{fig:usps_methods}.
It shows that the proposed method has comparable performance despite being localized. The performance is also affected by the choice of the variation operators (or, the GFT bases). Figure~\ref{fig:usps_bases} shows that the variation operators based on the hub-authority model and random walk offer higher classification accuracy and thus, are more suited for this particular application. Their superior performance can be explained by looking at the signal representation in the respective GFT domains. Figure~\ref{fig:usps_gft} shows the fraction of signal energy captured in increasing number of GFT coefficients starting from low frequency. Since the hub-authority model based GFT and random walk based GFT offer more energy compaction than adjacency based GFT, the signal reconstruction quality using these bases is naturally better.

\section{Conclusion}
\label{sec:conclusion}
We studied the problem of selecting an optimal sampling set for reconstruction of bandlimited graph signals. The starting point of our framework is the notion of the Graph Fourier Transform (GFT) which is defined via an appropriate variation operator. Our goal is to find good sampling sets for reconstructing signals which are bandlimited in the above frequency domain. We showed that when the samples are noisy or the true signal is only approximately bandlimited, the reconstruction error depends not only on the model mismatch but also on the choice of sampling set. We proposed a measure of quality for the sampling sets, namely the cutoff frequency, that can be computed without finding the GFT basis explicitly. A sampling set that maximizes the cutoff frequency is shown to minimize the reconstruction error. We also proposed a greedy algorithm which finds an approximately optimal set. The proposed algorithm can be efficiently implemented in a distributed and parallel fashion. Together with localized signal reconstruction methods, it gives an effective method for sampling and reconstruction of smooth graph signals on large graphs. 

The present work opens up some new questions for future research. The problem of finding a sampling set with maximum cutoff frequency is combinatorial. The proposed greedy algorithm gives only an approximate solution to this problem. It would be useful to find a polynomial time algorithm with theoretical guarantees on the quality of approximation. Further, the proposed set selection method is not adaptive, i.e., the choice of sampling locations does not depend on previously observed samples. This can be a limitation in applications that require  batch sampling. In such cases, it would be desirable to have an adaptive sampling set selection scheme which takes into account the previously observed samples to refine the choice of nodes to be sampled in the future.

\appendix[Properties of spectral proxies]
\label{app:spectral_moments}

In this section, we prove the monotonicity and convergence properties of $\omega_k(\fv)$.

\begin{proposition}
If $\Lm$ has real eigenvalues and eigenvectors, then for any $k_1 < k_2$, we have $\omega_{k_1}(\fv) \leq \omega_{k_2}(\fv), \forall \fv$.
\end{proposition}
\begin{proof}
We first expand $\omega_{k_1}(\fv)$ as follows:
\begin{align}
\left( \omega_{k_1}(\fv) \right)^{2k_1} &= \left( \frac{\| \Lm^{k_1} \fv\|}{\|\fv\|} \right)^{2} \nonumber \\
&= \frac{\sum_{i,j} (\lambda_i\lambda_j)^{k_1} \tilde{\fv}_i \tilde{\fv}_j \uv_i^\top \uv_j}{\sum_{i,j} \tilde{\fv}_i \tilde{\fv}_j \uv_i^\top \uv_j} \\
&= \sum_{i,j} (\lambda_i \lambda_j)^{k_1} c_{ij}
\end{align}
where $c_{ij} = \tilde{\fv}_i \tilde{\fv}_j \uv_i^\top \uv_j/\sum_{i,j} \tilde{\fv}_i \tilde{\fv}_j \uv_i^\top \uv_j$. Now, consider the function $f(x) = x^{k_2/k_1}$. Note that since $k_1 < k_2$, $f(x)$ is a convex function. Further, since $\sum_{i,j}c_{ij} = 1$, we can use Jensen's inequality in the above equation to get
\begin{align}
\left( \sum_{i,j} (\lambda_i \lambda_j)^{k_1} c_{ij} \right)^{k_2/k_1} &\leq \sum_{i,j} \left( (\lambda_i \lambda_j)^{k_1} \right)^{k_2/k_1} c_{ij} \\
\Rightarrow \left( \sum_{i,j} (\lambda_i \lambda_j)^{k_1} c_{ij} \right)^{1/2k_1} &\leq \left( \sum_{i,j} (\lambda_i \lambda_j)^{k_2} c_{ij} \right)^{1/2k_2} \nonumber \\
\Rightarrow \omega_{k_1}(\fv) &\leq \omega_{k_2}(\fv)
\end{align}
If $\Lm$ has real entries, but complex eigenvalues and eigenvectors, then these occur in conjugate pairs, hence, the above summation is real. However, in that case, $\omega_k(\fv)$ is not guaranteed to increase in a monotonous fashion, since $c_{ij}$'s are not real and Jensen's inequality breaks down.
\end{proof}

\begin{proposition}
Let $\omega(\fv)$ be the bandwidth of any signal $\fv$. Then, the following holds:
\begin{equation}
\omega(\fv) = \lim_{k\rightarrow \infty} \omega_k(\fv) = \lim_{k\rightarrow \infty} \left( \frac{\| \Lm^k \fv\|}{\|\fv\|} \right)^{1/k}
\end{equation}
\end{proposition}
\begin{proof}
We first consider the case when $\Lm$ has real eigenvalues and eigenvectors. Let $\omega(\fv) = \lambda_p$, then we have:
\begin{align}
\omega_k(\fv) &= \left( \frac{\sum_{i,j=1}^p (\lambda_i\lambda_j)^k \tilde{\fv}_i \tilde{\fv}_j \uv_i^\top \uv_j}{\sum_{i,j=1}^p \tilde{\fv}_i \tilde{\fv}_j \uv_i^\top \uv_j} \right)^{1/2k} \\
&= \lambda_p \left( c_{pp} + \sum_{(i,j)\neq(p,p) } \left(\frac{\lambda_i}{\lambda_p}\frac{\lambda_j}{\lambda_p}\right)^k c_{ij} \right)^{1/2k}
\end{align}
where $c_{ij} = \tilde{\fv}_i \tilde{\fv}_j \uv_i^\top \uv_j/\sum_{i,j} \tilde{\fv}_i \tilde{\fv}_j \uv_i^\top \uv_j$. Taking limits, it is easy to observe that the term in parentheses evaluates to 1. Hence, we have
\begin{equation}
\lim_{k\rightarrow \infty} \omega_k(\fv) = \lambda_p = \omega(\fv)
\end{equation}
Now, if $\Lm$ has complex eigenvalues and eigenvectors, then these have to occur in conjugate pairs since $\Lm$ has real entries. Hence, for this case, we do a similar expansion as above and take $|\lambda_p|$ out of the expression. Then, the limit of the remaining term is once again equal to 1.
\end{proof}


\ifCLASSOPTIONcaptionsoff
  \newpage
\fi



\bibliographystyle{IEEEtran}
\bibliography{refs}

%









\end{document}